\documentclass[12pt, draftclsnofoot, onecolumn,romanappendices]{IEEEtran}
\usepackage{bm,cite}
\usepackage{amsmath}
\usepackage{amsthm}
\usepackage{amsbsy}
\usepackage{epsfig}
\usepackage{latexsym}
\usepackage{amssymb}
\usepackage{wasysym}
\usepackage{placeins}

\newtheorem{theorem}{Theorem}
\newtheorem{proposition}{Proposition}
\newtheorem{lemma}{Lemma}
\newtheorem{corollary}{Corollary}
\DeclareMathAlphabet{\mathbit}{OML}{cmr}{bx}{it}
\DeclareMathAlphabet{\mathsf}{OT1}{cmss}{m}{n}
\DeclareMathAlphabet{\mathTXf}{OT1}{cmss}{bx}{it}

\DeclareMathOperator{\diag}{diag}

\DeclareMathOperator*{\argmax}{argmax}

\newcommand{\I}{\mathbf{I}} %identity matrix
\newcommand{\norm}[1]{\lVert{#1}\rVert}

\newcommand{\Fro}{{\text{F}}}

\newcommand{\MG}{{\text{M}_{\text{G}}}}

\newcommand{\E}{{\text{E}}}

\newcommand{\trans}{{\text{T}}} 
 
\newcommand{\He}{{{\text{H}}}}

\begin{document} 
\title{The Multiplexing Gain of a Two-cell MIMO Channel with Unequal CSI}
\author{ Paul de Kerret and David Gesbert\\Mobile Communications Department, Eurecom\\
 2229 route des Cr\^etes, 06560 Sophia Antipolis, France\\\{dekerret,gesbert\}@eurecom.fr}

\maketitle

\begin{abstract}
In this work\footnote{This work has been performed in the framework of the European research project ARTIST4G, which is partly funded by the European Union under its FP7 ICT Objective 1.1 - The Network of the Future.}, the joint precoding across two distant transmitters (TXs), sharing the knowledge of the data symbols to be transmitted, to two receivers (RXs), each equipped with one antenna, is discussed. We consider a distributed channel state information (CSI) configuration where each TX has its own local estimate of the channel and no communication is possible between the TXs. Based on the distributed CSI configuration, we introduce a concept of distributed MIMO precoding. We focus on the high signal-to-noise ratio (SNR) regime such that the two TXs aim at designing a precoding matrix to cancel the interference. Building on the study of the multiple antenna broadcast channel, we obtain the following key results: We derive the multiplexing gain (MG) as a function of the scaling in the SNR of the number of bits quantizing at each TX the channel to a given RX. Particularly, we show that the conventional Zero Forcing precoder is not MG maximizing, and we provide a precoding scheme optimal in terms of MG. Beyond the established MG optimality, simulations show that the proposed precoding schemes achieve better performances at intermediate SNR than known linear precoders.
\end{abstract}  
\IEEEpeerreviewmaketitle
\section{Introduction}
One promising solution to answer the need for increased spectral efficiency in the future wireless networks consists in the joint transmission from several transmitter (TXs) to serve multiple receivers (RXs), so called Network MIMO \cite{Karakayali2006,Somekh2006}. If all the TXs have access to the data symbols and to the global channel state information (CSI), the different TXs can then be seen as a unique virtual TX serving all the receivers (RXs). The precoding schemes of the multiple antenna broadcast channel (BC) can then be applied. 

Yet, this requires the sharing of the data symbol and the CSI between the TXs, which represents a high requirement on the network infrastructure. Indeed, while in future wireless networks (e.g. LTE Advanced), it is considered to link the TXs with the Core Network via high capacity links to share the data symbols with the cooperating TXs, the sharing of the CSI is done through limited rate feedback channels and limited capacity signaling (so called X2) links between the TXs. Thus, an interesting information theoretic MIMO channel arises whereby multiple TXs may access the same data symbols, but have a limited CSI sharing capability. We define this channel as the \emph{distributed CSI} (DCSI)-MIMO channel.

In the DCSI-MIMO channel, there may be inconsistencies between the different versions of CSI seen at the TXs due either to separate compression or separate feedback channels. Such inconsistencies can be detrimental to the channel capacity if they are not accounted for in the precoding design. This is the object of this work.

To put this in contrast, note that in the BC, the impact of finite rate feedback \cite{Jindal2006,Samardzija2006,Ding2007,Yoo2007} and the derivation of robust solutions\cite{Shenouda2006,Vucic2009} have been the focus of many works, which have been then extended to the MIMO network setting \cite{Kobayashi2009,Huh2010}. However, these works only focus on the case of imperfect CSI yet \emph{perfectly shared between the TXs} and do not consider the case when each TX has its own imperfect estimation of the multi-user channel, which will be our focus in this work. This setting was first studied in \cite{Zakhour2010}, and a tractable discrete optimization at finite SNR was derived. However, it does not lend itself to a more general performance analysis. 

Our work can be seen as a generalization to the case of distributed CSI setting of the study by Jindal~\cite{Jindal2006} of the multiple-antenna BC, in which the Multiplexing Gain (MG) is derived as a function of the number of feedback bits by each RX. We here consider only two TX-RX pairs, while the generalization to multiple TX-RX pairs is carried out in~\cite{dekerret2011_ISIT_journal}. We consider only \emph{Zero-Forcing} schemes which are known to achieve the maximal MG with perfect CSI in the MIMO BC. %Moreover, time multiplexing is also not considered.

Specifically, the main contributions are as follows. Let's first define the number of bits quantizing the estimate at TX~$j$ of the normalized channel~$\tilde{\bm{h}}_i^{\He}$ from the two TXs to RX~$i$ as $\alpha_i^{(j)}\log_2(P)$ with $\alpha_i^{(j)}\in[0,1]$. Then, we show that:
\begin{itemize}
\item The MG achieved with conventional Zero Forcing at RX~$i$ is equal to~$\min_{i,j\in\{1,2\}}\alpha_i^{(j)}$.
\item The optimal MG at RX~$i$ is equal to $\max_{j\in\{1,2\}}\alpha_i^{(j)}$.
\item We provide a precoding scheme achieving the maximal MG, as well as practical precoding schemes outperforming known linear precoding schemes at finite SNR for the DCSI-MIMO channel. 
\end{itemize}

\emph{Notations:} We denote by $\Pi_{\bm{u}}(\bullet)$ and $\Pi_{\bm{u}}^{\perp}(\bullet)$ the orthogonal projectors over the subspace spanned by $\bm{u}$ and over its orthogonal complement, respectively. $\bar{i}$ denotes the complementary indice of $i$, i.e., $\bar{i}=i\mod 2+1$.

\section{System Model}
We first present the classical multicell MIMO model before introducing our novel concepts of \emph{distributed CSI} and \emph{distributed precoding}.
\subsection{Multicell MIMO}
We consider a joint downlink transmission from two TXs to two RXs using linear precoding and single user decoding. For ease of exposition, the TXs and the RXs are equipped with only one antenna, such that the received signal is written as
\begin{equation}
\begin{bmatrix}
y_1\\y_2
\end{bmatrix}
\!=\!
\mathbf{H}\bm{x} 
\!+\!
\begin{bmatrix}
\eta_1\\
\eta_2
\end{bmatrix}
\!=\!
\begin{bmatrix}
\bm{h}^{\He}_1\bm{x}\\
\bm{h}^{\He}_2\bm{x}
\end{bmatrix}
\!+\!
\begin{bmatrix}
\eta_1\\
\eta_2
\end{bmatrix}
\!=\!
\begin{bmatrix}
\norm{\bm{h}^{\He}_1}\tilde{\bm{h}}^{\He}_1\bm{x}\\
\norm{\bm{h}^{\He}_2}\tilde{\bm{h}}^{\He}_2\bm{x}
\end{bmatrix}
\!+\!
\begin{bmatrix}
\eta_1\\
\eta_2
\end{bmatrix}
%\begin{bmatrix}
%h_{11}&h_{12}\\
%h_{21}&h_{22}
%\end{bmatrix}\bm{x} 
%+
%\begin{bmatrix}
%\eta_1\\
%\eta_2
%\end{bmatrix}
\label{eq:SM_1}
\end{equation}
where $y_i$ is the signal received at the $i$-th RX, $\bm{h}^{\He}_i \in \mathbb{C}^{1\times 2}$ is the channel from the TXs to the $i$-th RX, $\tilde{\bm{h}}^{\He}_i\triangleq \bm{h}^{\He}_i/\norm{\bm{h}^{\He}_i}$ is the normalized channel, $\eta_i\sim \mathcal{CN}(0,1)$ is the noise at the $i$-th RX and is distributed as i.i.d. complex circularly symmetric Gaussian noise, and $\bm{x}\in \mathbb{C}^{2\times 1}$ is the transmitted signal from the TXs. The channel is block fading and the entries of the channel matrix $\mathbf{H}$ are distributed as i.i.d. complex circularly symmetric Gaussian with unit variance to model a Rayleigh fading channel. The transmitted signal $\bm{x}$ is obtained from the vector of transmit symbol $\bm{s}=[s_1,s_2]^{\trans}\in\mathbb{C}^{2\times 1}$ (whose entries are assumed to be independent $\mathcal{CN}(0,1)$) as
\begin{equation}
\bm{x}=\mathbf{T}\bm{s}=
\begin{bmatrix}\bm{t}_1& \bm{t}_2\end{bmatrix}
\begin{bmatrix}s_1\\s_2\end{bmatrix} 
\label{eq:SM_2}
\end{equation}
where $\mathbf{T}\in \mathbb{C}^{2\times 2}$ and $\bm{t}_i \in \mathbb{C}^{2\times 1}$ is the beamforming vector used to transmit $s_i$. We consider a sum power constraint $\norm{\mathbf{T}}_{\Fro}^2=P$ and we also assume for simplicity and symmetry that $\bm{t}_i=\sqrt{P/2}\bm{u}_i$ with $\norm{\bm{u}_i}^{2}_2=1$. Note that normalizing the individual columns does not alter the ability to zero-force the interference so that it does not affect the MG. 

We also define the MG at RX~$i$ as
\begin{equation}
\MG_i\!\triangleq\!\lim_{P\rightarrow \infty}\!\frac{R_i(P)}{\log_2(P)}\!
\label{eq:SM_5}
\end{equation}
so that the total MG is $\MG\triangleq \MG_1+\MG_2$.

We will study the long-term average throughput over the fading distribution and also over the realizations of the Random Vector Quantization (RVQ) codebooks used for the CSI quantization (Cf. subsection~\ref{se:RVQ}), such that the throughput for RX~$i$ reads as
\begin{equation}
R_i(P)\triangleq\E_{\mathbf{H},\mathcal{W}}\left[\log_2\left(1+\frac{|\bm{h}_i^{\He}\bm{t}_i|^2}{1+|\bm{h}_{i}^{\He}\bm{t}_{\bar{i}}|^2}\right)\right]
\label{eq:SM_3}
\end{equation}
To achieve the maximal MG we aim at removing all the interference, i.e., at having
\begin{equation}
\mathcal{I}_1(\bm{t}_2)\triangleq|\bm{h}_1^{\He}\bm{t}_2|^2=0\text{, and  } \mathcal{I}_2(\bm{t}_1)\triangleq|\bm{h}_2^{\He}\bm{t}_1|^2=0.
\label{eq:SM_4}
\end{equation}
From \eqref{eq:SM_4}, it follows that the optimization of the two beamforming vectors~$\bm{t}_1$ and $\bm{t}_2$ can be uncoupled.

\subsection{Distributed CSI}

We assume a limited CSI setting where finite quality channel estimates are obtained from quantizing the true channel vectors. The \emph{distributed} CSI is defined here in the sense that each TX has a different estimate of the normalized channel $\tilde{\bm{h}}_i$ from all TXs to RX~$i$. Moreover, the estimates for $\tilde{\bm{h}}_1$ and $\tilde{\bm{h}}_2$ are also a priori of statistically different qualities. We denote by $\tilde{\bm{h}}^{(j)}_i$ the estimate of the normalized channel vector~$\tilde{\bm{h}}_i$ acquired at TX~$j$. Furthermore, the number of quantizing bits for $\tilde{\bm{h}}^{(j)}_i$ is given by~$B^{(j)}_i$.

In the context of MIMO BC, it is shown in~\cite{Jindal2006} that the number of quantization bits should scale indefinitely with the SNR in order to achieve a positive MG with ZF. It also holds in a distributed CSI configuration so that we focus on the \emph{scaling in the logarithm of the SNR} of the number of quantization bits
\begin{equation}
\alpha_i^{(j)}\triangleq\lim_{P\rightarrow \infty}\frac{B^{(j)}_i}{\log_2(P)}.
\label{eq:SM_6}
\end{equation}
Since $\alpha_i^{(j)}\!=\!1,\!\forall i,j\in\{1,2\}$ is shown later in Theorem~\ref{thm_MG_cZF} to be sufficient to achieve the maximal MG, we will always consider that $\alpha_i^{(j)}\in[0,1]$.

\subsection{Random Vector Quantization}\label{se:RVQ}
We consider the performances averaged over codebooks used to quantize the channels randomly generated. This follows a result in \cite{Jindal2006} stating that in the case of two antennas at the TX, no codebook can achieve in the single TX case a better MG than the MG achieved with RVQ. Moreover, RVQ is shown to be optimal as the number of antennas tends to infinity at the TX and the RXs~\cite{Santipach2009}.

However, in the MIMO BC, a codeword $\bm{c}$ is selected to quantize $\bm{h}$ if it maximizes the inner product $|\bm{h}^{\He}\bm{c}|$ over the codebook. Any other codeword of the form $\bm{c}e^{j\phi}$ where $\phi$ is any real number achieves the same performances and can be selected indifferently. This is problematic in a distributed setting since we are now interested in $\|\tilde{\bm{h}}^{(1)}_i-\tilde{\bm{h}}^{(2)}_i\|$ and even if the codewords at TX1 and TX2 are $e^{j\phi_1}\tilde{\bm{h}}_i$ and $e^{j\phi_2}\tilde{\bm{h}}_i$ respectively, i.e., exactly in the direction of $\tilde{\bm{h}}_i$, the two estimates differ greatly in norm.	 

Our solution is for each codeword and each channel estimate to choose $e^{j\phi}$ as the complex conjugate of the first vector element divided by its absolute value, thus making the first vector element real valued. Because of this choise, the quantization scheme is not any longer in the Grassmann manifold and we have to consider the isomorphisme between $\mathbb{C}$ and $\mathbb{R}^2$. Thus, for the quantization, each complex vector is considered as a vector of $\mathbb{R}^4$ made of the stacked real and imaginary parts. Moreover, since the first coefficient is real valued only, we have to consider in fact $\mathbb{R}^3$ only. A vector $\bm{u}\in\mathbb{C}^{2}$ with is first coefficient real valued is represented in $\mathbb{R}^3$ as $\bm{u}_{\mathbb{R}^{3}}$ and is defined as
\begin{equation}
\bm{u}_{\mathbb{R}^{3}}\triangleq\begin{bmatrix}
\mathrm{Re}(u_1)\\
\mathrm{Re}(u_2)\\
\mathrm{Im}(u_2)
\end{bmatrix} 
\end{equation}
Thus, we define the angles between $\bm{u}_{\mathbb{R}^{3}}$ and $\bm{v}_{\mathbb{R}^{3}}$ in $\mathbb{R}^{3}$ as 
\begin{equation}
\angle(\bm{u}_{\mathbb{R}^{3}},\bm{v}_{\mathbb{R}^{3}})=\arccos\left(\frac{|\bm{u}_{\mathbb{R}^{3}}^{\trans}\bm{v}_{\mathbb{R}^{3}}|}{\norm{\bm{u}_{\mathbb{R}^{3}}}\norm{\bm{v}_{\mathbb{R}^{3}}}}\right).
\label{eq:SM_6}
\end{equation}
Finally, the estimate $\tilde{\bm{h}}^{(j)}_i$ is chosen as the element of the random codebook~$\mathcal{W}$ which maximizes the cosinus of the angle between the codeword in $\mathbb{R}^{3}$ and the true channel in $\mathbb{R}^{3}$:
\begin{equation}
\tilde{\bm{h}}^{(j)}_{i\mathbb{R}^{3}}=\argmax_{\bm{c}_{\mathbb{R}^{3}}\in\mathcal{W}_{\mathbb{R}^{3}}} \cos(\angle(\bm{c}_{\mathbb{R}^{3}},\tilde{\bm{h}}_{i\mathbb{R}^{3}}))=\argmax_{\bm{c}_{\mathbb{R}^{3}}\in\mathcal{W}_{\mathbb{R}^{3}}}|\bm{c}_{\mathbb{R}^{3}}^{\trans}\tilde{\bm{h}}_{i\mathbb{R}^{3}}|.
\label{eq:SM_6}
\end{equation}

\subsection{Distributed Precoding}
In the distributed CSI setting, each TX has a different estimate of the channel, which it uses to compute the precoding matrix. We denote the overall multi-transmitter precoder computed at TX~$j$ as
\begin{equation}
\mathbf{T}^{(j)}\triangleq\begin{bmatrix}\bm{t}_1^{(j)}&\bm{t}_2^{(j)}\end{bmatrix}\triangleq\begin{bmatrix}T_{11}^{(j)}&T_{12}^{(j)}\\T_{21}^{(j)}&T_{22}^{(j)}\end{bmatrix}
\label{eq:SM_7}
\end{equation} 
where $\bm{t}_i^{(j)}$ is the beamforming vector transmitting $s_i$ computed at TX~$j$. Note that although a given TX~$j$ may compute the whole precoding matrix $\mathbf{T}^{(j)}$, only the $j$-th row will be used in practice since the other row corresponds to the coefficients being implemented at the other TX. Finally, the effective precoder is given by
\begin{equation}
\mathbf{T}\!=\!\begin{bmatrix}T_{11}^{(1)}\!&\! T_{12}^{(1)} \\T_{21}^{(2)}\!&\! T_{22}^{(2)}\end{bmatrix}.
\label{eq:SM_8}
\end{equation}

\section{Main Theorems on the Multiplexing Gain}
In the multiple antenna BC with perfect CSI, ZF achieves the maximal MG and can be conjectured to be also optimal with imperfect CSI. The central question of this paper is whether this result still holds in the DCSI-MIMO channel, and what are otherwise the MG optimal precoding strategies.

\subsection{Conventional Zero Forcing}
The conventional ZF precoder applied distributively consists in transmitting symbol~$i$ with the beamformer $\bm{t}_i^{\mathrm{cZF}}\triangleq[t_{1i}^{\mathrm{cZF}(1)},t_{2i}^{\mathrm{cZF}(2)}]^ {\trans}$, with its elements defined in an intuitive maneer as
\begin{equation} 
\bm{t}_i^{\mathrm{cZF}(j)}\triangleq \begin{bmatrix}t_{1i}^{\mathrm{cZF}(j)}\\t_{2i}^{\mathrm{cZF}(j)}\end{bmatrix}
\triangleq
\sqrt{\frac{P}{2}}\frac{
\Pi_{\tilde{\bm{h}}_{\bar{i}}^{(j)}}^{\perp}\left(\tilde{\bm{h}}_{i}^{(j)}\right)}{\norm{\Pi_{\tilde{\bm{h}}_{\bar{i}}^{(j)}}^{\perp}\left(\tilde{\bm{h}}_{i}^{(j)}\right)}},\quad j\in\{1,2\}.
\label{eq:def_cZF}
\end{equation}
The interpretation behind conventional ZF is that each TX applies ZF with its own CSI implicitely assuming that the other TX shares the same CSI estimate. Our first important result given in the following theorem relates the MG achieved with such a precoding strategy.

\begin{theorem}
Conventional ZF achieves the following MG: 
\begin{equation}
\MG^{c\mathrm{ZF}}=2\min_{i,j\in\{1,2\}}\alpha_i^{(j)}.
\label{eq:thm_cZF}
\end{equation}
\label{thm_MG_cZF}
\end{theorem}
\begin{IEEEproof}
A detailed proof is given in Appendix~\ref{se:proof_thm_MG_cZF}.
\end{IEEEproof} 
\begin{corollary}
Conventional ZF achieves the maximal MG if and only if the CSI scaling is identical across the RXs and the TXs, i.e., 
\begin{equation}
\forall i,j,\ell,k\in\{1,2\},\alpha_i^{(\ell)}=\alpha_k^{(j)}.
\end{equation} 
\label{corollary1}
\end{corollary}
\begin{IEEEproof}
The corollary follows from the comparison between the MG given in Theorem~\ref{thm_MG_cZF} and the MG achieved in a multiple antenna BC with imperfect CSI of the same quality \cite{Jindal2006}.
\end{IEEEproof}
It means that if the quality of the CSI is the same across all the TXs, it is in fact sufficient to apply conventional ZF. Even though it might seem a trivial result, it is not since additionnal error arise due to the fact that the estimates are not shared. The quality of the CSI is the same but estimates are different. This corollary also shows that the additionnal errors due to the CSI inconsistency do not lead to any further loss in MG.

\subsection{Robust Zero Forcing}
Comparing the MG in Theorem~\ref{thm_MG_cZF} and in a multiple antenna BC~\cite{Jindal2006}, it appears that in the case of imperfectly shared CSI, the MG is limited by the worst quality of the CSI across the channels to the RXs and across the TXs, which is a very pessimistic result. Robust precoding schemes have been derived in the literature either as statistical robust ZF precoder or precoder optimizing the worst case performances \cite{Shenouda2006} to reduce the harmful effect of the imperfect CSI. However, the robust versions improve the rate offset but do not have any impact on the MG.

%%%%%%%%%%%%%%
%%%% TO BE STUDIED
%%%%%%%%%%%%%%%
\subsection{Beacon Zero Forcing}
Robust ZF schemes from the literature do not bring any MG improvement. This leads us to investigate other schemes more adapted to the DCSI-MIMO channel. Thus, we now propose a modification of the conventional ZF scheme which improves the MG when the estimates for $\tilde{\bm{h}}_1$ and $\tilde{\bm{h}}_2$ are of different qualities. We call it \emph{Beacon ZF} (bZF) because it makes use of an arbitrary vector known at both TXs. 

The beamformer used to transmit symbol~$i$ is then $\bm{t}_i^{\mathrm{bZF}}\triangleq[t_{1i}^{\mathrm{bZF}(1)},t_{2i}^{\mathrm{bZF}(2)}]^ {\trans}$, with its elements defined as
\begin{equation}  
\bm{t}_i^{\mathrm{bZF}(j)}\triangleq
\begin{bmatrix}t_{1i}^{\mathrm{bZF}(j)}\\t_{2i}^{\mathrm{bZF}(j)} \end{bmatrix}\triangleq \sqrt{\frac{P}{2}}\frac{
\Pi_{\tilde{\bm{h}}_{\bar{i}}^{(j)}}^{\perp}\left(\bm{c}_{i}\right)}{\norm{\Pi_{\tilde{\bm{h}}_{\bar{i}}^{(j)}}^{\perp}\left(\bm{c}_{i}\right)}}
\label{eq:def_bZF}
\end{equation}
where~$\bm{c}_{i}$ is a vector chosen beforehand and known at the TXs. Due to the isotropy of the channel, the choice of~$\bm{c}_{i}$ is arbitrary and does not influence the performances of the precoder.

\begin{corollary}
The MG achieved with beacon ZF is 
\begin{equation}
\MG^{\mathrm{bZF}}=\min_{j\in\{1,2\}}\alpha_1^{(j)}+\min_{j\in\{1,2\}}\alpha_2^{(j)} 
\label{eq:thm_bZF}
\end{equation}
\label{corollary_MG_bZF_1}
\end{corollary}
\begin{proof}
The MG follows easily from Theorem~\ref{thm_MG_cZF}. Indeed, when using beacon ZF, no error is induced by the projection of the direct channel which is replaced by a fixed given vector. In terms of MG, there is no difference between projecting the direct channel or any given vector. Thus, we can 
\end{proof}

\begin{corollary}
Beacon ZF achieves the maximal MG if and only if $\forall i\in\{1,2\},\alpha_i^{(1)}=\alpha_i^{(2)}$. Thus, the inconsistency in the channel realizations between the TXs does not reduce the MG.
\label{corollary_MG_bZF_2}
\end{corollary}
\begin{proof}
Let consider that the TXs are allowed to cooperate, then any of the estimates can be used and the other thrown away. The channel used in the orthogonality constraint is then known with the given accuracy and no ZF can improve the accuracy. Thus, the maximal accuracy is achieved using beacon ZF.
\end{proof}
The key idea behind Beacon ZF is to reduce the impact of the difference in quality between $\bm{h}_1$ and $\bm{h}_2$ by using only the CSI necessary to fulfill the orthogonality constraint and not the direct channel which does not change the MG but only improves the finite SNR performances. Indeed, $\bm{t}_1^{\mathrm{bZF}}$ does not depend on the estimates of $\bm{h}_1$ and symmetrically $\bm{t}_2^{\mathrm{bZF}}$ does not depend on the estimates of $\bm{h}_2$. 
%Note that ZF brings an improvement of the rate offset compared to Beacon ZF only when the TXs have more than one antenna, since in that case classical ZF leads to chose a vector in the nullspace of $\bm{h}_{\bar{i}}$ giving a good gain at RX~$i$. With only two TXs,  the nullspace is of dimension one and classical ZF can only perform worse than Beacon ZF.

\subsection{Active/Passive - Zero Forcing}
Beacon ZF improves the MG in some settings but it is still the worst CSI quality across the TXs which implies the MG. Thus, we now propose a scheme called \emph{Active/Passive} Zero Forcing (A/P-ZF) to take care of this problem. Assuming wlog that $\alpha_{\bar{i}}^{(2)}\geq \alpha_{\bar{i}}^{(1)}$, it consists in the precoder whose beamformer to transmit symbol~$i$ is given by 
\begin{equation}
\bm{t}_i^{\mathrm{A/P-ZF}}\!\triangleq\!\sqrt{\!\frac{P(1\!+\!\rho_{i}^{(2)})}{2\log_2(P)}}\bm{u}_i^{\mathrm{apZF}}\!\triangleq\!\sqrt{\frac{P}{2\log_2(P)}}\!\begin{bmatrix}1\\
-\frac{\tilde{h}_{\bar{i}1}^{(2)}}{\tilde{h}_{\bar{i}2}^{(2)}}\!
\end{bmatrix} 
\label{eq:def_apZF}
\end{equation}
where $\tilde{\bm{h}}_{\bar{i}}^{(2)\He}\!\triangleq\![\tilde{h}_{\bar{i}1}^{(2)},\tilde{h}_{\bar{i}2}^{(2)}]$, $\rho_{i}^{(2)}\!\triangleq \!|\tilde{h}_{\bar{i} 1}^{(2)}|^2/|\tilde{h}_{\bar{i} 2}^{(2)}|^2$, $\norm{\bm{u}_i^{\mathrm{A/P-ZF}}}=1$.

A/P-ZF is based on the idea that each beamforming vector has to fullfill one orthogonality constraint so that only degree of freedom is necessary. Thus, one coefficient can be set to a constant while still fullfilling the ZF constraints. Additionnaly, the other underlying idea is that the only way to achieve the MG steming from the best CSI estimate is if TX~$2$ (which has the best knowledge of  $\tilde{\bm{h}}_1$) can adapt to the transmission done at TX~$1$ to adjust its beamforming vector and improves how the interference are suppressed. This is possible only if TX~$2$ knows the coefficient used to transmit at TX~$1$ which means that TX~$1$ should not use its own CSI and transmit with a fixed coefficient. The MG using this precoding scheme is then given in the following proposition.

\begin{proposition} 
Active-Passive ZF achieves the MG:
\begin{equation}
\MG^{\mathrm{A/P-ZF}}\geq\max_{j\in[1,2]}\alpha_1^{(j)}+\max_{j\in[1,2]}\alpha_2^{(j)}.
\label{eq:thm:apZF}
\end{equation}
\label{proposition}
\end{proposition}

\begin{proof}
Due to the symmetry between the two RXs, we consider only the MG at RX~$1$, and we consider that the beamformers $\bm{t}_1$ and $\bm{t}_2$ are given by \eqref{eq:def_apZF}. We still assume wlog that $\alpha_{1}^{(2)}\geq \alpha_{1}^{(1)}$, i.e., TX~$2$ has the best CSI over $\tilde{\bm{h}}_1$. Using A/P-ZF, the MG at RX~$1$ reads as
\begin{align}
\MG_1&\!=\! \!\lim_{P\rightarrow\infty}\!\!\!\frac{\E_{\mathbf{H},\mathcal{W}}\left[\log_2\left(1\!+\!\frac{\norm{\bm{h}_1}^2\norm{\bm{t}_1}^2|\tilde{\bm{h}}_1^{\He}\bm{u}_1|^2}{\sigma_1^2+\mathcal{I}_1(\bm{t}_{2})}\!\right)\!\right]}{\log_2(P)}\nonumber\\
&\!=\! \lim_{P\rightarrow \infty}\frac{\E_{\mathbf{H},\mathcal{W}}\left[\log_2\left(\frac{(\rho_{2}^{(2)}+1)P}{\log_2(P)}\right)-\log_2(\mathcal{I}_1(\bm{t}_2))\right]}{\log_2(P)}\nonumber\\
%\MG_1&\!\!\!=\! \lim_{P\rightarrow \infty}\frac{\E_{\mathbf{H},\mathcal{W}}\left[\log_2(P)\!-\!\log_2(\log_2(P))\!-\!\log_2(\mathcal{I}_1(\bm{t}_{2}))\right]}{\log_2(P)}\nonumber\\
&\!= 1-\lim_{P\rightarrow \infty}\frac{\E_{\mathbf{H},\mathcal{W}}\left[\log_2(\mathcal{I}_1(\bm{t}_{2}))\right]}{\log_2(P)}\\
&\!\geq 1-\lim_{P\rightarrow \infty}\frac{\log_2\left(\E_{\mathbf{H},\mathcal{W}}\left[\mathcal{I}_1(\bm{t}_{2})\right]\right)}{\log_2(P)}
\label{eq:proofapZF_1}
\end{align}
where we have used Jensen's inequality for the last inequality. We now consider the interference term~$\mathcal{I}_1(\bm{t}_{2})$:
\begin{equation}
\mathcal{I}_1(\bm{t}_2)=|\bm{h}_1^{\He}\bm{t}_2|^2=\frac{P}{2\log_2(P)}\left|\bm{h}_1^{\He}\begin{bmatrix}1\\ -\frac{\tilde{h}_{11}^{(2)}}{\tilde{h}_{12}^{(2)}}\end{bmatrix}\right|^2.
\end{equation}
By construction, $\bm{t}_2$ is orthogonal to $\bm{h}_1^{(2)}$, so that 
\begin{align}
\mathcal{I}_1(\bm{t}_2)&=\frac{P(1+\rho_{2}^{(2)})}{2\log_2(P)}\left|\left(\Pi_{\bm{h}_1^{(2)}}^{\perp}(\bm{h}_1)+\Pi_{\bm{h}_1^{(2)}}(\bm{h}_1)\right)^{\He}\bm{u}_2 \right|^2\nonumber\\
&=\frac{P(1+\rho_{2}^{(2)})}{2\log_2(P)}\norm{\bm{h}_1}^2 \left|\Pi_{\bm{h}_1^{(2)}}^{\perp}(\bm{h}_1)^{\He}\bm{u}_2 \right|^2\\
&\leq\frac{P(1+\rho_{2}^{(2)})}{2\log_2(P)}\norm{\bm{h}_1}^2 \norm{\Pi_{\bm{h}_1^{(2)}}^{\perp}(\bm{h}_1)}^2\\
&\leq\frac{P(1+\rho_{2}^{(2)})}{2\log_2(P)}\norm{\bm{h}_1}^2 \sin^2(\angle(\bm{h}_1,\bm{h}_1^{(2)})).
\label{eq:proofapZF_3}
\end{align} 
Inserting \eqref{eq:proofapZF_3} into the MG expression \eqref{eq:proofapZF_1} and using Proposition~\ref{App_distorsion} of Appendix~\ref{se:Appendix_RVQ}, we obtain
\begin{equation}
\begin{aligned} 
%\MG_1&=  1-\lim_{P\rightarrow \infty}\frac{\E_{\mathbf{H},\mathcal{W}}\left[\log_2(\mathcal{I}_1(\bm{t}_2))\right]}{\log_2(P)}\\
\MG_1&\geq\lim_{P\rightarrow \infty} \frac{-E_{\mathbf{H},\mathcal{W}}\left[\log_2(|\Pi_{\bm{h}_1^{(2)}}^{\perp}(\bm{h}_1)^{\He}\bm{u}_2|^2)\right]}{\log_2(P)}\\ 
\MG_1&\geq\lim_{P\rightarrow \infty} \frac{B_1^{(2)}}{\log_2(P)}= \alpha_1^{(2)} 
\end{aligned}
\label{eq:proofapZF_5}
\end{equation}
which is the best scaling across the TXs.
\end{proof}
Comparing the MG achieved with A/P-ZF with the MG achieved when both TXs are allowed to exchange their channel estimates, the following fundamental result follows directly.
\begin{theorem} 
Active/Passive ZF achieves the maximal MG.
\label{theorem4}
\end{theorem}
\begin{proof}
Let assume that the sharing of the channel estimates is allowed between the TXs. Then it is optimal to use the best estimates for each of the channel vector and simply throw the other estimate. In that case, the TXs share the same CSI quality and it is optimal to apply Beacon-ZF which achieves the same MG as given in \ref{corollary_MG_bZF_1}. This MG is an upper bound for the MG so that A/P-ZF achieves the maximal MG.
\end{proof}

A/P-ZF allows to recover the MG which would have been achieved with the sharing of the estimates. However, one last point remains to be discussed: the choice of the coefficient used to transmit at TX~$1$. Actually, the beamformer can be multiplied arbitrarily by any unit norm complex number without impacting the rate achieved, so that only the power used at TX~$1$ needs to be decided. In \eqref{eq:def_cZF}, the power used is set to $P/(2\log_2(P))$, which follows the fact that the channel $h_{22}$ might have a very small amplitude, in which case it would be necessary for TX~$2$ to transmit with a very large power to cancel the interference. To ensure that the interference are canceled for all channel realizations while respecting the power constraint, it is necessary to have the ratio between the power used at TX~$1$ and the total power tending to zero. The factor $\log_2(P)$ is used because it fulfills this property while not reducing the MG due to the partial power consumption.

\subsection{Power Control for A/P-ZF Precoding}

We have seen that A/P-ZF could achieve a much better MG than conventional ZF. However, this comes at the cost of using only a small share of the available power, which is clearly inefficient and leads to bad performances at finite SNR. To improve the performances, the TX with the worst accuracy needs to adapt its power consumption to the channel realizations. In the following, we propose two possible solutions.
\begin{itemize}
\item 
Firstly, TX~$1$ can use its local CSI to normalize the beamformer which is then given by
\begin{equation}
\bm{t}_i^{\mathrm{apZF}}=\sqrt{\frac{P}{2}}\begin{bmatrix}\frac{1}{\sqrt{1+\rho_i^{(1)}}}\\
-\frac{\tilde{h}_{\bar{i}1}^{(2)}}{\sqrt{1+\rho_i^{(2)}}\tilde{h}_{\bar{i}2}^{(2)}}
\end{bmatrix} 
\label{eq:def_hapZF}
\end{equation}
with $\rho_{i}^{(j)}\triangleq |\tilde{h}_{\bar{i} 1}^{(j)}|^2/|\tilde{h}_{\bar{i} 2}^{(j)}|^2$. This beamformer is not MG maximizing because the local CSI is used at TX~$1$ so that TX~$2$ cannot adapt to it to cancel the interference, and the beamformer is not orthogonal to $\tilde{\bm{h}}_{\bar{i}}^{(2)}$. Yet, this solution achieves good performance at intermediate SNR.RVQ

\item Another possibility is to assume that TX~$1$ receives the scalar $\rho_i^ {(2)}$ (or $\rho_i$) and use it to control its power. This means that either the RX or TX~$2$ needs to feedback this scalar. It requires an additionnal feedback, but only a few bits are necessary, because it is only used to improve the power efficiency and does not impact the MG. Thus, the feedback of this scalar does not change the scaling of the CSI in terms of the SNR nor the performances, and appears thus as an interesting practical solution. 
\end{itemize}
\section{Simulations}
We consider two models for the imperfect channel CSI, a statistical model and RVQ. In the statistical model, the quantization error is modeled by adding a Gaussian i.i.d. quantization noise to the channel with the covariance matrix at TX~$j$ equal to $\diag([1/P^{\alpha_1^{(j)}},1/P^{\alpha_2^{(j)}}])$. When considering given finite number of feedback bits, we compute $\alpha_i^{(j)}=B_i^{(j)}/\log_2(P)$, so that $\diag([1/P^{\alpha_1^{(j)}},1/P^{\alpha_2^{(j)}}])=\diag([1/2^{B_1^{(j)}},1/2^{B_2^{(j)}}])$. For RVQ, we consider a number of quantizing bits either numerically given or obtained from the CSI scaling as $q_i^{(j)}=\lfloor \alpha_i^{(j)}\log_2(P)\rfloor$. In the statistical model, we average over $10000$ realization and for RVQ we average over $100$ codebooks and $1000$ channel realizations. In the simulations. we consider the following precoders: ZF with perfect CSI, conventional ZF [cf.~\eqref{eq:def_cZF}], beacon ZF [cf.~\eqref{eq:def_bZF}], and Active/Passive ZF [cf.~\eqref{eq:def_apZF}] with heuristic power control and with $4$-bits power control.

In Fig.~\ref{figure1}, we consider the statistical model with the CSI scaling $[\alpha_1^{(1)},\alpha_1^ {(2)}]=[1,0.5]$ and $[\alpha_2^{(1)},\alpha_2^{(2)}]=[0,0.7]$. To emphasize the MG (i.e., the slope of the curve in the figure), we let the SNR grow large. As expected theoretically, conventional ZF saturates at high SNR, while Beacon ZF has a positive slope and Active/Passive ZF performs close to perfect ZF with a slope only slightly smaller than the optimal one.

\begin{figure}%[htp!] 
\centering
\includegraphics[width=1\columnwidth]{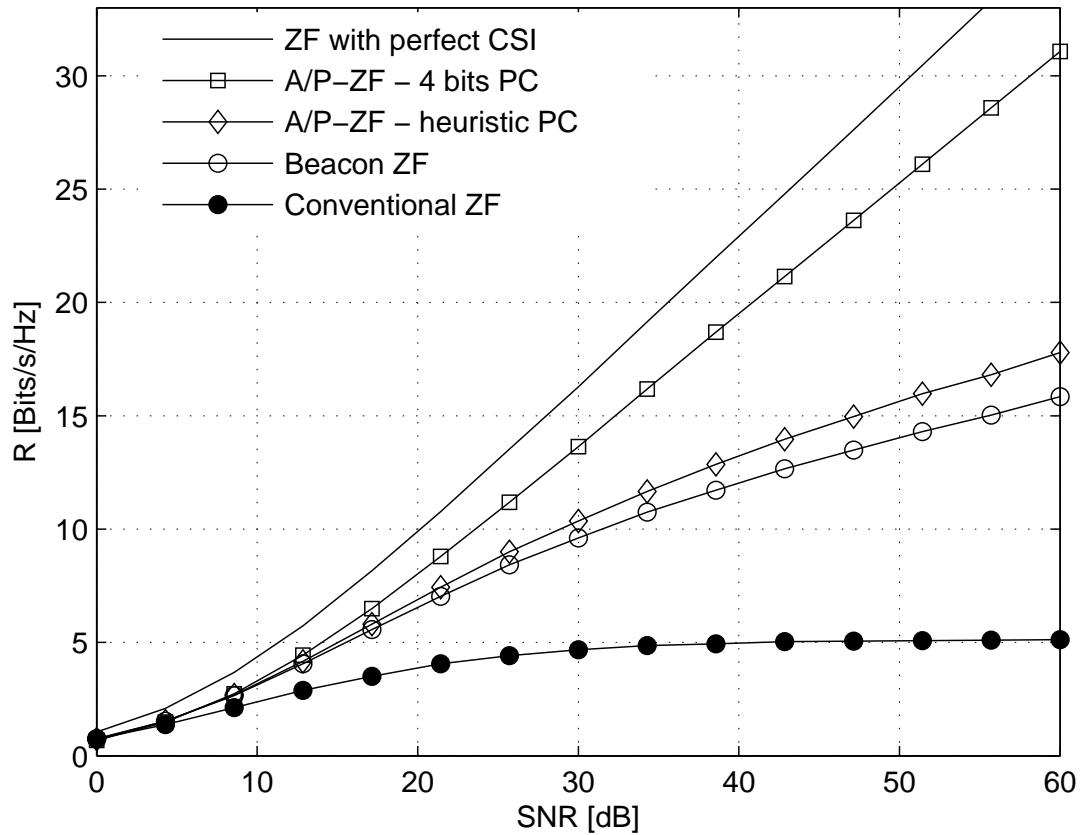}
\caption{Sum rate in terms of the SNR with a statistical modeling of the error from RVQ using $[\alpha_1^ {(1)},\alpha_1^ {(2)}]=[1,0.5]$ and $[\alpha_2^{(1)},\alpha_2^{(2)}]=[0,0.7]$.}
\label{figure1}
\end{figure}

In Fig.~\ref{figure2} and  Fig.~\ref{figure3} we plot the sum rate achieved with the CSI feedback $[B_1^{(1)},B_1^{(2)}]=[6,3]$ and $[B_2^{(1)},B_2^{(2)}]=[3,6]$ for the statistical modeling and RVQ, respectively. Firstly, we can observe the good match between the two models used. From the theoretical analysis the MG is null for all the precoding schemes for a finite number of feedback bits, which can be observed by the saturation of the sum rate as the SNR grows. Yet, the saturation occurs at higher SNR for Beacon $ZF$ compared to conventional ZF, and at even higher SNR for Active/Passive-ZF, which leads to an improvement of the sum rate even at intermediate SNR. %Finally, cooperative ZF with heuristic power control seems to achieve the same MG as generalized ZF but performs better at intermediate SNR.

\begin{figure}%[htp!] 
\centering
\includegraphics[width=1\columnwidth]{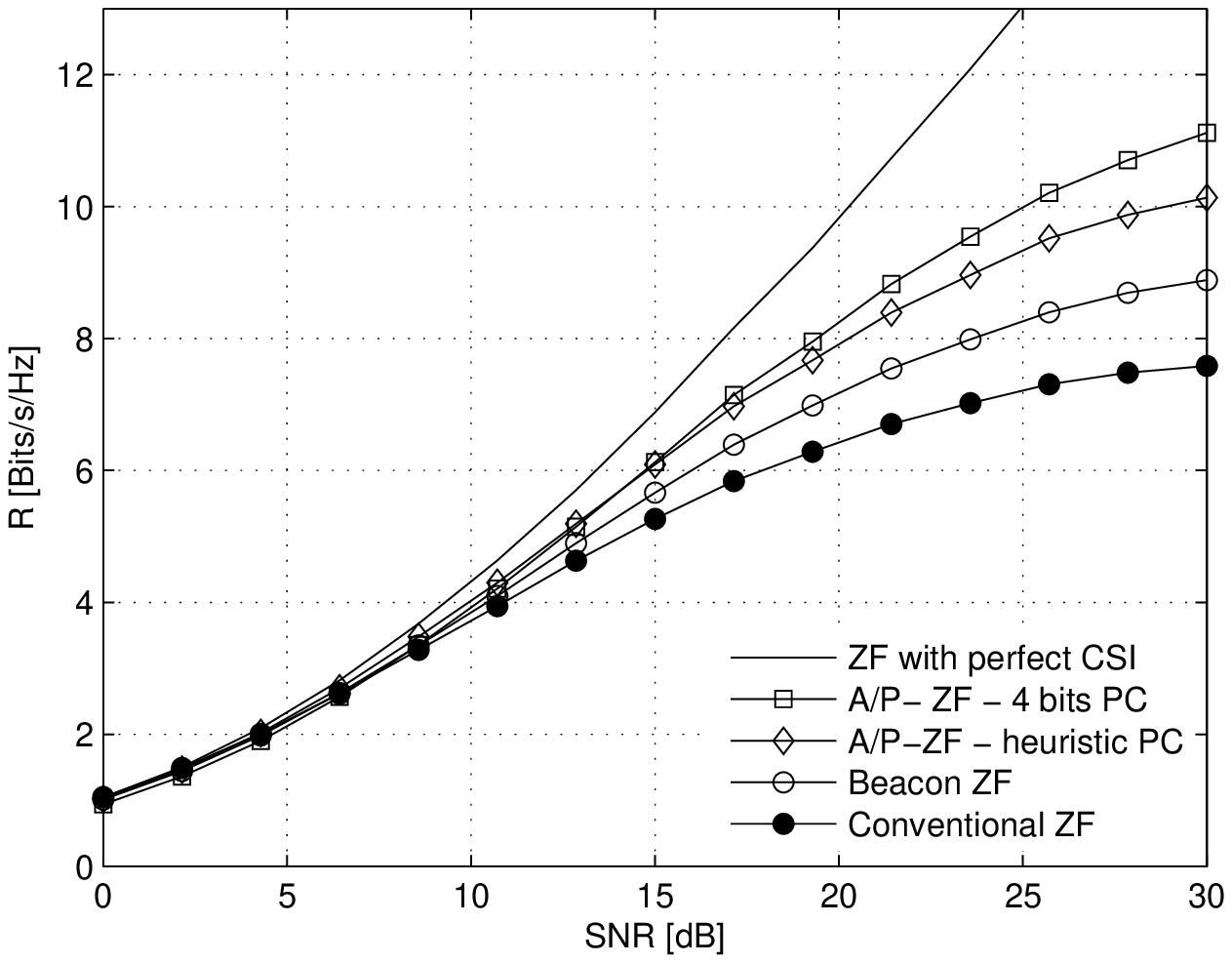}
\caption{Sum rate in terms of the SNR with a statistical modeling of the error obtained using $[B_1^{(1)},B_1^{(2)}]=[6,3]$ and $[B_2^{(1)},B_2^{(2)}]=[3,6]$.}
\label{figure2}
\end{figure}

\begin{figure}%[htp!] 
\centering
\includegraphics[width=1\columnwidth]{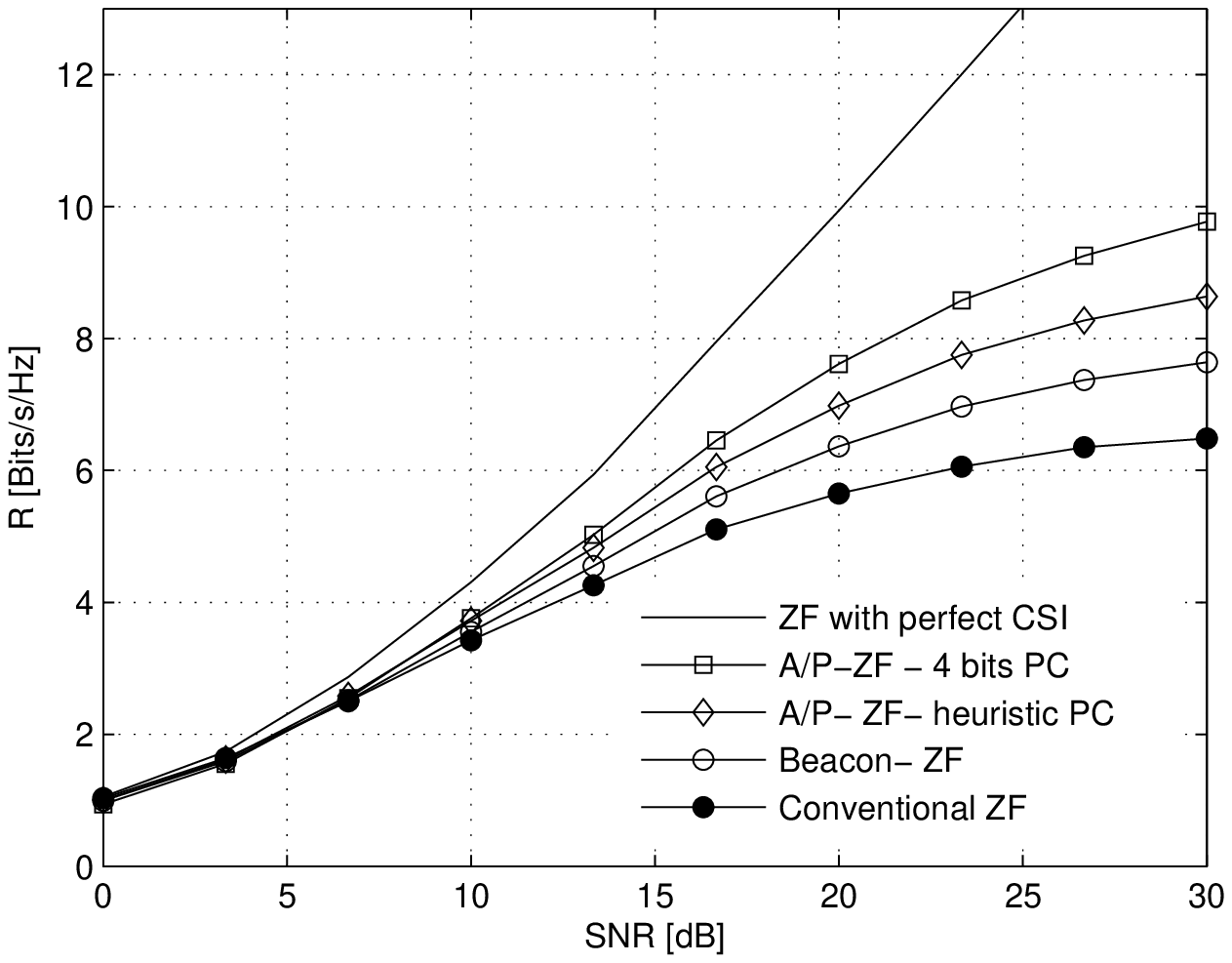}
\caption{Sum rate in terms of the SNR with RVQ using $[B_1^{(1)},B_1^{(2)}]=[6,3]$ and $[B_2^{(1)},B_2^{(2)}]=[3,6]$.}
\label{figure3}
\end{figure}

\section{Conclusion}
In this work, the multiplexing gain in a two-cell broadcast channel where the TXs have different estimates of the multi-user channel has been studied. We have shown that usual Zero Forcing precoding applied without taking into account the differences in CSI quality achieves far from the maximal MG. We have also derived the value of the maximal MG in that distributed CSI configuration and provided a MG maximizing precoding scheme. Moreover, we have shown by simulations that the new precoding approach outperforms known linear precoding schemes at intermediate SNR. We have considered only two TXs and two RXs with a single antenna to keep the notations simple, but the extension to multiple-antenna TXs or RXs appears to be tractable while the analysis in the case of $K$ TX-RX pairs with a single antenna is done in \cite{dekerret2011_ISIT_journal}. %Finally, the combination of the TX cooperation in the distributed CSI configuration with a scheduler represents also an interesting direction of research.
%

 %In fact every precoding scheme presented here could be extended in future works to a robust version so as to improve the rate offset at finite SNR, but with no impact on the MG.

\FloatBarrier
\section{Appendix}
\subsection{Some Results on Vector Quantization}\label{se:Appendix_RVQ}
In this section, we recall some results on vector quantizations from~\cite{Dai2008} and we derive some new properties which will be needed for the following proofs. We consider the angle between two vectors as defined in \eqref{eq:SM_6}. It means that we are considering real vectors of unit norm in the linear space $\mathbb{R}^{2n-1}$ as explained in Subsection~\ref{se:RVQ}.
\begin{proposition}[\cite{Dai2008}, Corollary~$2$]
The cumulative distribution function (CDF) of $\mathrm{d}^2(\tilde{\bm{h}},\bm{c})\triangleq\sin^2(\angle(\tilde{\bm{h}},\bm{c}))$ where $\bm{c}$ is an element of a random codebook is bounded as	
\begin{equation} 
c_{2n-1}x^{n-1}\leq\mathrm{F}(x)\triangleq\mathrm{Pr}\{\sin^2(\angle(\tilde{\bm{h}},\bm{c}))\leq x\}\leq c_{2n-1}x^{n-1}(1-x)^{\frac{-1}{2}}.
\label{eq:App_1}
\end{equation}
where $\tilde{\bm{h}}$ is the true unitary channel
\label{App_CDF}
\end{proposition}

\begin{proposition}[\cite{Dai2008}, Theorem~$2$]
When the size $K=2^B$ of the random codebook is sufficiently large ($c_{2n-1}^{-1/(n-1)}2^{-B/(n-1)}\leq 1$ necessary), then it holds that
\begin{equation} 
\E_{\mathcal{C},\tilde{\bm{h}}}[\min_{\bm{c}\in\mathcal{C}}\sin^2(\angle(\tilde{\bm{h}},\bm{c}))]\leq \frac{\Gamma(\frac{1}{n-1})}{n-1}c_{2n-1}^{-1/(n-1)}2^{-B/(n-1)}(1+o(1))
\label{eq:App_2}
\end{equation}
where $\tilde{\bm{h}}$ is the true channel and $c_{2n-1}\triangleq \Gamma(n-1/2)/(\Gamma(n)\Gamma(1/2))$.
\label{App_distorsion}
\end{proposition}

\begin{proposition}
The expectation of the logarithm of the quantization error is bounded as
\begin{equation} 
\frac{B+\log_2(c_{2n-1})}{(n-1)}\leq\E_{\mathcal{C},\tilde{\bm{h}}}\left[-\log_2\left(\min_{\bm{c}\in\mathcal{C}}\sin^2(\angle(\tilde{\bm{h}},\bm{c}))\right)\right]\leq\frac{B+\log_2(c_{2n-1})+\log_2(e)}{(n-1)}
\label{eq:App_3}
\end{equation} 
\label{App_log_distorsion}
where $\hat{\bm{h}}$ is the best estimate over a random codebook of size $2^B$, $\tilde{\bm{h}}$ is the true channel and $c_{2n-1}\triangleq \Gamma(n-1/2)/(\Gamma(n)\Gamma(1/2))$. 
\end{proposition}
\begin{proof} 
\emph{Upper Bound:} The derivation of an upper bound follows the same idea as the proof in Appendix~$B$ of \cite{Dai2008} which derives an upper bound for the expectation as considered here, but without the logarithm. We start by recalling a Lemma from \cite{Dai2008} which follows easily from the definition but brings some insight.
\begin{lemma}[\cite{Dai2008}, Lemma~$3$]
The empirical distribution function minimizing the distorsion over a given $K$ is 
\begin{equation}
 \mathrm{F}^{*}_{\mathcal{C}^{*}}(x)\triangleq=
     \begin{cases}
        0 & \text{if $x<0$} \\
        K \mathrm{F}(x)& \text{if $0\leq x\leq x^{*}$} \\
        1 & \text{if $x>x^{*}$}
     \end{cases}     
\end{equation}
where $x^{*}$ satisfies $K\mathrm{F}(x^{*})=1$ and $\mathrm{F}(x)\triangleq\mathrm{Pr}\{\sin^2(\angle(\tilde{\bm{h}},\bm{c}))|\leq x\}$ is the CDF of the squared distance between a random vector and one element of the random codebook.
\label{App_lemma_CDF}
\end{lemma}
Note that Lemma~\ref{App_lemma_CDF} corresponds to the optimal codebook minimizing the average distance and correponds thus to a lower bound for the distorsion. We define $Z\triangleq\sin^2(\angle(\tilde{\bm{h}},{\bm{c}}))$ and use the fact that the term considered in the expectation is a positive random variable to rewrite the expectation in function of its CDF.
\begin{equation}
\begin{aligned}
\E_{\mathcal{C},\tilde{\bm{h}}}\left[-\log\left(\min_{\bm{c}\in\mathcal{C}}\sin^2(\angle(\tilde{\bm{h}},\bm{c}))\right)\right]
&=\int_0^{\infty}\mathrm{Pr}\{-\log\left(\min_{\bm{c}\in\mathcal{C}}\sin^2(\angle(\tilde{\bm{h}},\bm{c}))\right)\geq z\}\mathrm{d}z\\
&=\int_0^{\infty}\mathrm{Pr}\{\min_{\bm{c}\in\mathcal{C}}\sin^2(\angle(\tilde{\bm{h}},\bm{c}))\leq e^ {-z}\}\mathrm{d}z\\
&\geq\int_0^{\infty}K\mathrm{Pr}\{Z\leq e^ {-z}\}\mathrm{d}z.
\end{aligned}
\label{eq_App_log_distorsion_1}
\end{equation}
Following the same approach as the proof in Appendix~$B$ of \cite{Dai2008}, we define $\mathrm{F}_0(x)\triangleq c_{2n-1} x^{n-1}$ and $x_0$ so that $K\mathrm{F}_0(x_0)=1$. Let also define $\mathrm{F}_{\mathrm{ub}}(x)\triangleq c_{2n-1} x^{n-1}(1-x)^{-1/2}$ and $x_{\mathrm{ub}}$ so that $K\mathrm{F}_{\mathrm{ub}}(x_{\mathrm{ub}})=1$. Finaly. we define $\mathrm{F}_{\mathrm{ubub}}(x)\triangleq c_{2n-1} x^{n-1}(1-x_0)^{-1/2}$ and $x_{\mathrm{ubub}}$ so that $K\mathrm{F}_{\mathrm{ubub}}(x_{\mathrm{ubub}})=1$.

It holds by construction that $x_{\mathrm{ub}}\leq x^{*}\leq x_0$ since we know from Proposition~\ref{App_CDF} that $\mathrm{F}_0(x)\leq \mathrm{F}(x)\leq \mathrm{F}_{\mathrm{ub}}(x)$. Thus, $(1-x)^{-1/2}\leq(1-x_0)^{-1/2}$ for $x\in [0,x_{\mathrm{ub}}]$ so that $\mathrm{F}_{\mathrm{ub}}(x)\leq \mathrm{F}_{\mathrm{ubub}}(x)$ for $x\in [0,x_{\mathrm{ub}}]$, which finally implies $x_{\mathrm{ubub}}\leq x_{\mathrm{ub}}$. We can then use these relations to upper bound \eqref{eq_App_log_distorsion_1}.
\begin{equation}
\begin{aligned}
\E_{\mathcal{C},\tilde{\bm{h}}}\left[-\log\left(\min_{\bm{c}\in\mathcal{C}}\sin^2(\angle(\tilde{\bm{h}},\bm{c}))\right)\right]
&\leq\int_0^{\infty}K\mathrm{Pr}\{Z\leq e^ {-z}\}\mathrm{d}z\\
&\leq\int_0^{-\log(x^{*})}\mathrm{d}z+\int_{-\log(x^{*})}^{\infty}K\mathrm{F}(e^ {-z})\}\mathrm{d}z\\
&\leq\int_0^{-\log(x^{*})}\mathrm{d}z+\int_{-\log(x^{*})}^{\infty}K\mathrm{F}_{\mathrm{ub}}(e^ {-z})\}\mathrm{d}z\\
&=\int_0^{-\log(x_{\mathrm{ub}})}\mathrm{d}z+\int_{-\log(x_{\mathrm{ub}})}^{\infty}K\mathrm{F}_{\mathrm{ub}}(e^{-z})\}\mathrm{d}z\\
&\leq\int_0^{-\log(x_{\mathrm{ub}})}\mathrm{d}z+\int_{-\log(x_{\mathrm{ub}})}^{\infty}K\mathrm{F}_{\mathrm{ubub}}(e^{-z})\}\mathrm{d}z\\
&=\int_0^{-\log(x_{\mathrm{ubub}})}\mathrm{d}z+\int_{-\log(x_{\mathrm{ubub}})}^{\infty}K\mathrm{F}_{\mathrm{ubub}}(e^{-z})\}\mathrm{d}z.
\end{aligned}
\label{eq_App_log_distorsion_1}
\end{equation}
We now replace $\mathrm{F}_{\mathrm{ubub}}$ and $x_{\mathrm{ubub}}$ by their expressions and evaluate the integral.
\begin{equation}
\begin{aligned}
\E_{\mathcal{C},\tilde{\bm{h}}}\left[-\log\left(\min_{\bm{c}\in\mathcal{C}}\sin^2(\angle(\tilde{\bm{h}},\bm{c}))\right)\right]
&\leq -\frac{1}{n-1}\log\left(\frac{(1-x_0)^{1/2}}{Kc_{2n-1}}\right)+\frac{Kc_{2n-1}}{(1-x_0)^{1/2}}\int_{-\log(x_{\mathrm{ubub}})}^{\infty}e^{-z(n-1)}
\mathrm{d}z\\
&= -\frac{1}{n-1}\log\left(\frac{(1-(Kc_{2n-1})^{\frac{-1}{n-1}})^{1/2}}{Kc_{2n-1}}\right)+\frac{1}{n-1}\\
&\approx \frac{1}{n-1}\left(\log\left(Kc_{2n-1}\right)+1\right)
\end{aligned}
\label{eq_App_log_distorsion_2}
\end{equation}
for $K$ large enough. Dividing by $\log(2)$ yields the final upper bound.

\emph{Lower Bound:} To derive the lower bound, we use the lower bound for the CDF given in Proposition~\ref{App_CDF}. The CDF has a form very simililar to the CDF for the single TX case considering complex quantization so that we can adapt the approach of the proof of Lemma~$3$ by Jindal in \cite{Jindal2006} to the current setting. Defining also $Z\triangleq\sin^2(\angle(\tilde{\bm{h}},{\bm{c}}))$, the CDF is given by $\mathrm{Pr}\{Z\geq z\}=1-c_{2n-1} z^{(n-1)}$ in Proposition~\ref{App_CDF}. We can use this CDF to derive that
\begin{equation}
\mathrm{Pr}\{\min_{\bm{c}\in\mathcal{C}}\sin^2(\angle(\tilde{\bm{h}},\bm{c}))\geq z\}=1-(1-c_{2n-1}. x^{(n-1)})^{K}
\end{equation}.
A lower bound for the expectation of the logarithm can then be calculated as follows.
\begin{equation} 
\begin{aligned}
\E_{\mathcal{C},\tilde{\bm{h}}}\left[-\log\left(\min_{\bm{c}\in\mathcal{C}}\sin^2(\angle(\tilde{\bm{h}},\bm{c}))\right)\right]
&\geq\int_0^{\infty}1-(1-c_{2n-1} e^{-z(n-1)})^{K}\mathrm{d}z\\ 
&\geq\int_0^{\infty}1-\sum_{k=0}^{K}\binom{K}{k}(-1)^{k} c_{2n-1}^k e^{-z(n-1)k}\mathrm{d}z\\
&\geq\frac{1}{n-1}\sum_{k=1}^{K}\binom{K}{k}(-1)^{k+1} \frac{c_{2n-1}^k}{k}
=\frac{1}{n-1}f(K)
\end{aligned}
\label{eq_App_log_distorsion_3}
\end{equation}
where we have defined $f(p)\triangleq\sum_{k=1}^{p}\binom{p}{k}(-1)^{k+1} \frac{c_{2n-1}^k}{k}$. To compute the value of $f(K)$, we will use the following relation given in \cite[Sec. $0.155$]{Gradshteyn2007}.
\begin{equation} 
\sum_{k=0}^n\binom{n}{k}\frac{\alpha^{k+1}}{k+1}=\frac{(\alpha+1)^{n+1}-1}{n+1}.
\label{eq_App_log_distorsion_4}
\end{equation}
We will now rewrite $f(K)$ in order to be able to apply \eqref{eq_App_log_distorsion_4}
\begin{equation} 
\begin{aligned}
f(K)
&\triangleq\sum_{k=1}^{K}\binom{K}{k}(-1)^{k+1} \frac{c_n^k}{k}\\
&= (-1)^{K+1} \frac{c_n^K}{K}+\sum_{k=1}^{K-1}\left[\binom{K-1}{k-1}+\binom{K-1}{k}\right](-1)^{k+1} \frac{c_{2n-1}^k}{k}\\
&=\sum_{k=1}^{K}\binom{K-1}{k-1}(-1)^{k+1}\frac{c_{2n-1}^{k}}{k}+\sum_{k=1}^{K-1}\binom{K-1}{k}(-1)^{k+1} \frac{c_{2n-1}^k}{k}\\
&=\sum_{k'=0}^{K-1}\binom{K-1}{k'}(-1)^{k'+2}\frac{c_{2n-1}^{k'+1}}{k'+1}+f(K-1)\\
&=-\frac{(-c_{2n-1}+1)^{K}-1}{K}+f(K-1)\\
&=\sum_{p=1}^{K}\frac{1-(-c_{2n-1}+1)^{p}}{p}\\
&=\sum_{p=1}^{K}\frac{1}{p}-\sum_{p=1}^{K}\frac{1-(-c_{2n-1}+1)^{p}}{p}.
\end{aligned}
\label{eq_App_log_distorsion_5}
\end{equation}
Furthermore we have the two following relations:
\begin{equation} 
\begin{aligned}
&\log(K)\leq \sum_{p=1}^{K}\frac{1}{p}\leq \log(K)+1\\
&\log(1-x)=-\sum_{k=1}^{\infty}\frac{x^k}{k}\text{ , for $x\in[-1,2]$}.
\end{aligned}
\label{eq_App_log_distorsion_6}
\end{equation}
Using these properties and dividing by $\log(2)$, we can obtain the final lower bound as
\begin{equation} 
\begin{aligned}
\E_{\mathcal{C},\tilde{\bm{h}}}\left[-\log\left(\min_{\bm{c}\in\mathcal{C}}\sin^2(\angle(\tilde{\bm{h}},\bm{c}))\right)\right]
&\geq \frac{1}{(n-1)\log(2)}\sum_{p=1}^{K}\frac{1}{p}-\frac{1}{(n-1)\log(2)}\sum_{p=1}^{K}\frac{(1-c_{2n-1})^{p}}{p}\\
&\geq \frac{\log_2(K)}{(n-1)}-\frac{1}{(n-1)\log(2)}\sum_{p=1}^{\infty}\frac{(1-c_{2n-1})^{p}}{p}\\
&= \frac{\log_2(K)+\log_2(c_n)}{(n-1)}.
\end{aligned}
\label{eq_App_log_distorsion_7}
\end{equation}

\end{proof}

\subsection{Proof of Theorem~\ref{thm_MG_cZF}}\label{se:proof_thm_MG_cZF}
\begin{proof}
The proof generalizes the proof of Theorem~$4$ in Appendix IV of \cite{Jindal2006}, which derives the multiplexing gain for the single TX case, to the distributed CSI configuration. However, an important difference, which makes the derivation more intricated, consists in the fact that we are not only interested in the inner product between the estimate and the true channel but also in the difference in norm between the estimates at the two TXs. This is the reason why we need to consider Grassmanian beamforming in $\mathbb{R}^{2n-1}$ instead of $\mathbb{C}^{n}$, as already introduced in Subsection~\ref{se:RVQ}.

We start by deriving two lemmas which form in fact the core of the proof.
\begin{lemma}
Let the beamformers $\bm{u}_2^{(1)}$ and $\bm{u}_2^{(2)}$ be computed at TX~$1$ and TX~$2$ respectively, then
\begin{equation}
\left\|\bm{u}_{2}^{(2)}-\bm{u}_{2}^{(1)}\right\|^2\leq C_{\mathrm{UB}} \max_{i=1,2,j=1,2}\left(\sin^2(\angle(\tilde{\bm{h}}_i^{(j)},\tilde{\bm{h}}_i))\right)
\end{equation}
where $C_{\mathrm{UB}}$ is some positive constant.
\label{lemma_UB}
\end{lemma}
\begin{proof}
Since the norm is conserved when considering the isomorphisme between $\mathbb{C}$ and $\mathbb{R}^2$, we work in fact in~$\mathbb{R}^2$ for the rest of the proof, even though we keep the same notations. The beamformer difference can then be rewritten as
\begin{equation} 
\left\|\bm{u}_{2}^{(2)}-\bm{u}_{2}^{(1)}\right\|
\!=\!\left\|\frac{(\I_n-\tilde{\bm{h}}_1^{(2)}\tilde{\bm{h}}_1^{(2)\trans})\tilde{\bm{h}}_2^{(2)}}{\norm{(\I_n-\tilde{\bm{h}}_1^{(2)}\tilde{\bm{h}}_1^{(2)\trans})\tilde{\bm{h}}_2^{(2)}}}-\frac{(\I_n-\tilde{\bm{h}}_1^{(1)}\tilde{\bm{h}}_1^{(1)\trans})\tilde{\bm{h}}_2^{(1)}}{\norm{(\I_n-\tilde{\bm{h}}_1^{(1)}\tilde{\bm{h}}_1^{(1)\trans})\tilde{\bm{h}}_2^{(1)}}}\right\|.
\label{eq:App_UB_1}
\end{equation}
We now decompose the estimates at TX~$2$ over the estimates at TX~$1$ as
{\begin{equation} 
\begin{aligned}
\tilde{\bm{h}}_1^{(2)}&=\Pi_{\tilde{\bm{h}}_1^{(1)}}^{\perp}(\tilde{\bm{h}}_1^{(2)})+(\tilde{\bm{h}}_1^{(1)\trans}\tilde{\bm{h}}_1^{(2)})\tilde{\bm{h}}_1^{(1)}\\
\tilde{\bm{h}}_2^{(2)}&=\Pi_{\tilde{\bm{h}}_2^{(1)}}^{\perp}(\tilde{\bm{h}}_2^{(2)})+(\tilde{\bm{h}}_2^{(1)\trans}\tilde{\bm{h}}_2^{(2)})\tilde{\bm{h}}_2^{(1)}.
\end{aligned}
\label{eq:App_UB_2}
\end{equation}}
Using \eqref{eq:App_UB_2}, we can write $(\tilde{\bm{h}}_1^{(2)\trans}\tilde{\bm{h}}_2^{(2)})$ as
{\begin{equation} 
\begin{aligned}
(\tilde{\bm{h}}_1^{(2)\trans}\tilde{\bm{h}}_2^{(2)})&=\left(\Pi_{\tilde{\bm{h}}_1^{(1)}}^{\perp}(\tilde{\bm{h}}_1^{(2)})+(\tilde{\bm{h}}_1^{(1)\trans}\tilde{\bm{h}}_1^{(2)})\tilde{\bm{h}}_1^{(1)}\right)^{\trans}\tilde{\bm{h}}_2^{(2)}\\ &=((\Pi_{\tilde{\bm{h}}_1^{(1)}}^{\perp}(\tilde{\bm{h}}_1^{(2)}))^{\trans}\tilde{\bm{h}}_2^{(2)})+(\tilde{\bm{h}}_1^{(1)\trans}\tilde{\bm{h}}_1^{(2)})(\tilde{\bm{h}}_1^{(1)\trans}\Pi_{\tilde{\bm{h}}_2^{(1)}}^{\perp}(\tilde{\bm{h}}_2^{(2)}))+(\tilde{\bm{h}}_1^{(1)\trans}\tilde{\bm{h}}_1^{(2)})(\tilde{\bm{h}}_1^{(1)\trans}\tilde{\bm{h}}_2^{(1)})(\tilde{\bm{h}}_2^{(1)\trans}\tilde{\bm{h}}_2^{(2)}).
\end{aligned}
\label{eq:App_UB_3}
\end{equation}}
In a first step, inserting only \eqref{eq:App_UB_3} in \eqref{eq:App_UB_1}, we can obtain the upper bound
{\small\begin{equation} 
\begin{aligned}
&\left\|\bm{u}_{2}^{(2)}-\bm{u}_{2}^{(1)}\right\|\!\\
&=\left\| \frac{\tilde{\bm{h}}_2^{(2)}-\left[((\Pi_{\tilde{\bm{h}}_1^{(1)}}^{\perp}(\tilde{\bm{h}}_1^{(2)}))^{\trans}\tilde{\bm{h}}_2^{(2)})+(\tilde{\bm{h}}_1^{(1)\trans}\tilde{\bm{h}}_1^{(2)})(\tilde{\bm{h}}_1^{(1)\trans}\Pi_{\tilde{\bm{h}}_2^{(1)}}^{\perp}(\tilde{\bm{h}}_2^{(2)}))+(\tilde{\bm{h}}_1^{(1)\trans}\tilde{\bm{h}}_1^{(2)})(\tilde{\bm{h}}_1^{(1)\trans}\tilde{\bm{h}}_2^{(1)})(\tilde{\bm{h}}_2^{(1)\trans}\tilde{\bm{h}}_2^{(2)})\right]\tilde{\bm{h}}_1^{(2)}}{\norm{(\I_n-\tilde{\bm{h}}_1^{(2)}\tilde{\bm{h}}_1^{(2)\trans})\tilde{\bm{h}}_2^{(2)}}}   -\bm{u}_{2}^{(1)}\right\|\\
&\leq\left\| \frac{\tilde{\bm{h}}_2^{(2)}-\left[(\tilde{\bm{h}}_1^{(1)\trans}\tilde{\bm{h}}_1^{(2)})(\tilde{\bm{h}}_1^{(1)\trans}\tilde{\bm{h}}_2^{(1)})(\tilde{\bm{h}}_2^{(1)\trans}\tilde{\bm{h}}_2^{(2)})\right]\tilde{\bm{h}}_1^{(2)}}{\norm{(\I_n-\tilde{\bm{h}}_1^{(2)}\tilde{\bm{h}}_1^{(2)\trans})\tilde{\bm{h}}_2^{(2)}}}   -\bm{u}_{2}^{(1)}\right\|+\varepsilon_1+\varepsilon_2
\end{aligned}
\label{eq:App_UB_4}
\end{equation}}
where we have defined 
{\begin{equation} 
\begin{aligned}
\varepsilon_1&\triangleq|(\Pi_{\tilde{\bm{h}}_1^{(1)}}^{\perp}(\tilde{\bm{h}}_1^{(2)}))^{\trans}\tilde{\bm{h}}_2^{(2)}|/\norm{(\I_n-\tilde{\bm{h}}_1^{(2)}\tilde{\bm{h}}_1^{(2)\trans})\tilde{\bm{h}}_2^{(2)}}\\
\varepsilon_2&\triangleq|(\tilde{\bm{h}}_1^{(1)\trans}\tilde{\bm{h}}_1^{(2)})\tilde{\bm{h}}_1^{(1)\trans}\Pi_{\tilde{\bm{h}}_2^{(1)}}^{\perp}(\tilde{\bm{h}}_2^{(2)})|/\norm{(\I_n-\tilde{\bm{h}}_1^{(2)}\tilde{\bm{h}}_1^{(2)\trans})\tilde{\bm{h}}_2^{(2)}}.
\end{aligned} 
\end{equation}}
We derive now an upper bound for $\varepsilon_1$ and the same approach will hold also for $\varepsilon_2$ and for other similar expressions later in the proof.
{\small\begin{equation} 
\begin{aligned}
\varepsilon_1\triangleq\frac{|(\Pi_{\tilde{\bm{h}}_1^{(1)}}^{\perp}(\tilde{\bm{h}}_1^{(2)})^{\trans}\tilde{\bm{h}}_2^{(2)}|}{\norm{(\I_n-\tilde{\bm{h}}_1^{(2)}\tilde{\bm{h}}_1^{(2)\trans})\tilde{\bm{h}}_2^{(2)}}}&\leq \frac{\norm{\Pi_{\tilde{\bm{h}}_1^{(1)}}^{\perp}(\tilde{\bm{h}}_1^{(2)})}}{\norm{(\I_n-\tilde{\bm{h}}_1^{(2)}\tilde{\bm{h}}_1^{(2)\trans})\tilde{\bm{h}}_2^{(2)}}}\\
&= \frac{\sqrt{1-|\tilde{\bm{h}}_1^{(1)\trans}\tilde{\bm{h}}_1^{(2)}|^2}}{\norm{(\I_n-\tilde{\bm{h}}_1^{(2)}\tilde{\bm{h}}_1^{(2)\trans})\tilde{\bm{h}}_2^{(2)}}}\\
&= \frac{|\sin(\angle(\tilde{\bm{h}}_1^{(1)},\tilde{\bm{h}}_1^{(2)}))|}{\norm{(\I_n-\tilde{\bm{h}}_1^{(2)}\tilde{\bm{h}}_1^{(2)\trans})\tilde{\bm{h}}_2^{(2)}}}.
\end{aligned}
\label{eq:App_UB_5}
\end{equation}}
We further use \eqref{eq:App_UB_2} to rewrite $\tilde{\bm{h}}_1^{(2)}$ and $\tilde{\bm{h}}_2^{(2)}$ and obtain an upper bound for the first term of the right hand side of \eqref{eq:App_UB_4} that we denote by $A$.
{\small\begin{equation} 
\begin{aligned}
	A&\triangleq\left\| \frac{\tilde{\bm{h}}_2^{(2)}-\left[(\tilde{\bm{h}}_1^{(1)\trans}\tilde{\bm{h}}_1^{(2)})(\tilde{\bm{h}}_1^{(1)\trans}\tilde{\bm{h}}_2^{(1)})(\tilde{\bm{h}}_2^{(1)\trans}\tilde{\bm{h}}_2^{(2)})\right]\tilde{\bm{h}}_1^{(2)}}{\norm{(\I_n-\tilde{\bm{h}}_1^{(2)}\tilde{\bm{h}}_1^{(2)\trans})\tilde{\bm{h}}_2^{(2)}}}   -\bm{u}_{2}^{(1)}\right\|\\
&=\left\| \frac{\Pi_{\tilde{\bm{h}}_2^{(1)}}^{\perp}(\tilde{\bm{h}}_2^{(2)})+(\tilde{\bm{h}}_2^{(1)\trans}\tilde{\bm{h}}_2^{(2)})\tilde{\bm{h}}_2^{(1)}-\left[(\tilde{\bm{h}}_1^{(1)\trans}\tilde{\bm{h}}_1^{(2)})(\tilde{\bm{h}}_1^{(1)\trans}\tilde{\bm{h}}_2^{(1)})(\tilde{\bm{h}}_2^{(1)\trans}\tilde{\bm{h}}_2^{(2)})\right]\left(\Pi_{\tilde{\bm{h}}_1^{(1)}}^{\perp}(\tilde{\bm{h}}_1^{(2)})+(\tilde{\bm{h}}_1^{(1)\trans}\tilde{\bm{h}}_1^{(2)})\tilde{\bm{h}}_1^{(1)}\right)}{\norm{(\I_n-\tilde{\bm{h}}_1^{(2)}\tilde{\bm{h}}_1^{(2)\trans})\tilde{\bm{h}}_2^{(2)}}}   -\bm{u}_{2}^{(1)}\right\|\\
&\leq \left\| \frac{(\tilde{\bm{h}}_2^{(1)\trans}\tilde{\bm{h}}_2^{(2)})\tilde{\bm{h}}_2^{(1)}-\left[(\tilde{\bm{h}}_1^{(1)\trans}\tilde{\bm{h}}_1^{(2)})(\tilde{\bm{h}}_1^{(1)\trans}\tilde{\bm{h}}_2^{(1)})(\tilde{\bm{h}}_2^{(1)\trans}\tilde{\bm{h}}_2^{(2)})\right](\tilde{\bm{h}}_1^{(1)\trans}\tilde{\bm{h}}_1^{(2)})\tilde{\bm{h}}_1^{(1)}}{\norm{(\I_n-\tilde{\bm{h}}_1^{(2)}\tilde{\bm{h}}_1^{(2)\trans})\tilde{\bm{h}}_2^{(2)}}}   -\bm{u}_{2}^{(1)}\right\|+\varepsilon_3+\varepsilon_4
\end{aligned}
\label{eq:App_UB_6}
\end{equation}}
where we have further defined
{\begin{equation} 
\begin{aligned}
\varepsilon_3&\triangleq \|\Pi_{\tilde{\bm{h}}_2^{(1)}}^{\perp}(\tilde{\bm{h}}_2^{(2)})\|/\norm{(\I_n-\tilde{\bm{h}}_1^{(2)}\tilde{\bm{h}}_1^{(2)\trans})\tilde{\bm{h}}_2^{(2)}}\\
\varepsilon_4&\triangleq\|\Pi_{\tilde{\bm{h}}_1^{(1)}}^{\perp}(\tilde{\bm{h}}_1^{(2)})\|/\norm{(\I_n-\tilde{\bm{h}}_1^{(2)}\tilde{\bm{h}}_1^{(2)\trans})\tilde{\bm{h}}_2^{(2)}}.
\end{aligned} 
\end{equation}}
Both $\varepsilon_3$ and $\varepsilon_4$ can be handled similarly to $\varepsilon_1$ in \eqref{eq:App_UB_5}. The norm in the denominator is now rewritten in terms of the estimations at TX~$2$ and approximated by using that the difference between the estimates is small compared to one since the accuracy of the CSI is increasing with the SNR.
{\small\begin{equation} 
\begin{aligned}
\frac{1}{\norm{(\I_n-\tilde{\bm{h}}_1^{(2)}\tilde{\bm{h}}_1^{(2)\trans})\tilde{\bm{h}}_2^{(2)}}}&=
\frac{1}{\norm{\tilde{\bm{h}}_2^{(1)}-(\tilde{\bm{h}}_1^{(1)\trans}\tilde{\bm{h}}_2^{(1)})\tilde{\bm{h}}_1^{(1)}+\bm{\varepsilon}}}\\
&\approx\frac{1}{\norm{\tilde{\bm{h}}_2^{(1)}-(\tilde{\bm{h}}_1^{(1)\trans}\tilde{\bm{h}}_2^{(1)})\tilde{\bm{h}}_1^{(1)}}}\left(1-\frac{\bm{\varepsilon}^{\trans}(\tilde{\bm{h}}_2^{(1)}-(\tilde{\bm{h}}_1^{(1)}\tilde{\bm{h}}_1^{(1)\trans})\tilde{\bm{h}}_2^{(1)})}{\norm{\tilde{\bm{h}}_2^{(1)}-(\tilde{\bm{h}}_1^{(1)}\tilde{\bm{h}}_1^{(1)\trans})\tilde{\bm{h}}_2^{(1)}}}\right)\\
\end{aligned}
\label{eq:App_UB_7}
\end{equation}}
where we have introduced
\begin{equation}
\bm{\varepsilon}\triangleq\tilde{\bm{h}}_2^{(2)}-(\tilde{\bm{h}}_1^{(2)\trans}\tilde{\bm{h}}_2^{(2)})\tilde{\bm{h}}_1^{(2)}-(\tilde{\bm{h}}_2^{(1)}-(\tilde{\bm{h}}_1^{(1)\trans}\tilde{\bm{h}}_2^{(1)})\tilde{\bm{h}}_1^{(1)}). 
\end{equation}
The term with $\bm{\varepsilon}$ can be upper bound using the triangular inequality and further upper bounded by the product of the norm. Finally, the norm of $\bm{\varepsilon}$ can be upper bounded by using the same steps that have been used for the numerator, i.e., by projecting the estimates at TX~$2$ over the estimates at TX~$1$ and upper bounding the terms. Using the result of the side calculation in \eqref{eq:App_UB_7}, we can now write an asymptotic bound for the first term in the right hand side of \eqref{eq:App_UB_6} which we denote by $B$.
{\small\begin{equation} 
\begin{aligned}
B&\triangleq \left\| \frac{(\tilde{\bm{h}}_2^{(1)\trans}\tilde{\bm{h}}_2^{(2)})\tilde{\bm{h}}_2^{(1)}-\left[(\tilde{\bm{h}}_1^{(1)\trans}\tilde{\bm{h}}_1^{(2)})(\tilde{\bm{h}}_1^{(1)\trans}\tilde{\bm{h}}_2^{(1)})(\tilde{\bm{h}}_2^{(1)\trans}\tilde{\bm{h}}_2^{(2)})\right](\tilde{\bm{h}}_1^{(1)\trans}\tilde{\bm{h}}_1^{(2)})\tilde{\bm{h}}_1^{(1)}}{\norm{\tilde{\bm{h}}_2^{(1)}-(\tilde{\bm{h}}_1^{(1)\trans}\tilde{\bm{h}}_2^{(1)})\tilde{\bm{h}}_1^{(1)}+\bm{\varepsilon}}}   -\frac{\tilde{\bm{h}}_2^{(1)}-(\tilde{\bm{h}}_1^{(1)\trans}\tilde{\bm{h}}_2^{(1)})\tilde{\bm{h}}_1^{(1)}}{\norm{\tilde{\bm{h}}_2^{(1)}-(\tilde{\bm{h}}_1^{(1)\trans}\tilde{\bm{h}}_2^{(1)})\tilde{\bm{h}}_1^{(1)}}}\right\|\\
&\leq \left\| (\tilde{\bm{h}}_2^{(1)\trans}\tilde{\bm{h}}_2^{(2)})\tilde{\bm{h}}_2^{(1)}-\left[(\tilde{\bm{h}}_1^{(1)\trans}\tilde{\bm{h}}_1^{(2)})(\tilde{\bm{h}}_1^{(1)\trans}\tilde{\bm{h}}_2^{(1)})(\tilde{\bm{h}}_2^{(1)\trans}\tilde{\bm{h}}_2^{(2)})\right](\tilde{\bm{h}}_1^{(1)\trans}\tilde{\bm{h}}_1^{(2)})\tilde{\bm{h}}_1^{(1)}   -\left(\tilde{\bm{h}}_2^{(1)}-(\tilde{\bm{h}}_1^{(1)\trans}\tilde{\bm{h}}_2^{(1)})\tilde{\bm{h}}_1^{(1)}\right)\right\|\\
&\leq \left| (\tilde{\bm{h}}_2^{(1)\trans}\tilde{\bm{h}}_2^{(2)})-1\right|+\left|\left(1-(\tilde{\bm{h}}_1^{(1)\trans}\tilde{\bm{h}}_1^{(2)})^2(\tilde{\bm{h}}_2^{(1)\trans}\tilde{\bm{h}}_2^{(2)})\right)(\tilde{\bm{h}}_1^{(1)\trans}\tilde{\bm{h}}_2^{(1)})\right|\\
&=\left|\cos(\angle(\tilde{\bm{h}}_2^{(1)},\tilde{\bm{h}}_2^{(2)}))\right|+\left|1-(\cos(\angle(\tilde{\bm{h}}_1^{(1)},\tilde{\bm{h}}_1^{(2)})))^2\cos(\angle(\tilde{\bm{h}}_2^{(1)},\tilde{\bm{h}}_2^{(2)})))\right|\left|\tilde{\bm{h}}_1^{(1)\trans}\tilde{\bm{h}}_2^{(1)}\right|\\
&=\left|2\sin^2(\angle(\tilde{\bm{h}}_2^{(1)},\tilde{\bm{h}}_2^{(2)})/2)\right|+\left|1-(1-2\sin^2(\angle(\tilde{\bm{h}}_1^{(1)},\tilde{\bm{h}}_1^{(2)})/2))^2 (1-2\sin^2(\angle(\tilde{\bm{h}}_2^{(1)},\tilde{\bm{h}}_2^{(2)})/2))\right|\left|\tilde{\bm{h}}_1^{(1)\trans}\tilde{\bm{h}}_2^{(1)}\right|\\
&\approx\left|\frac{\sin^2(\angle(\tilde{\bm{h}}_2^{(1)},\tilde{\bm{h}}_2^{(2)}))}{\cos^2(\angle(\tilde{\bm{h}}_2^{(1)},\tilde{\bm{h}}_2^{(2)}))}\right|+\left|2\frac{\sin^2(\angle(\tilde{\bm{h}}_1^{(1)},\tilde{\bm{h}}_1^{(2)}))}{\cos^2(\angle(\tilde{\bm{h}}_1^{(1)},\tilde{\bm{h}}_1^{(2)}))} +\frac{\sin^2(\angle(\tilde{\bm{h}}_2^{(1)},\tilde{\bm{h}}_2^{(2)}))}{\cos^2(\angle(\tilde{\bm{h}}_2^{(1)},\tilde{\bm{h}}_2^{(2)}))}\right|\left|\tilde{\bm{h}}_1^{(1)\trans}\tilde{\bm{h}}_2^{(1)}\right|.
\end{aligned}
\label{eq:App_UB_8}
\end{equation}}
Putting all the pieces together, we have shown that 
\begin{equation} 
\left\|\bm{u}_{2}^{(2)}-\bm{u}_{2}^{(1)}\right\|\!
\leq\! a_1 |\sin(\angle(\tilde{\bm{h}}_1^{(1)},\tilde{\bm{h}}_1^{(2)}))|+a_2 |\sin(\angle(\tilde{\bm{h}}_2^{(1)},\tilde{\bm{h}}_2^{(2)}))|
\label{eq:App_UB_9}
\end{equation}
which we now want to relate to the inner product with the true channels $\tilde{\bm{h}}_1$ and $\tilde{\bm{h}}_2$. Wlog we focus on the term $|\sin(\angle(\tilde{\bm{h}}_1^{(2)},\tilde{\bm{h}}_1^{(1)}))|$.
\begin{equation} 
\begin{aligned}
|\sin(\angle(\tilde{\bm{h}}_1^{(2)},\tilde{\bm{h}}_1^{(1)}))|&=\sqrt{1-|\tilde{\bm{h}}_1^{(2)\trans}\tilde{\bm{h}}_1^{(1)}|^2}\\
&=\|\Pi_{\tilde{\bm{h}_1}^{(1)}}^{\perp}(\tilde{\bm{h}}_1^{(2)})\|\\
&=\|\Pi_{\tilde{\bm{h}_1}^{(1)}}^{\perp}(\Pi_{\tilde{\bm{h}_1}}^{\perp}(\tilde{\bm{h}}_1^{(2)}))+\Pi_{\tilde{\bm{h}_1}^{(1)}}^{\perp}((\tilde{\bm{h}}_1^{(2)\trans}\tilde{\bm{h}}_1)\tilde{\bm{h}}_1)\|\\
&=\|\Pi_{\tilde{\bm{h}_1}^{(1)}}^{\perp}(\Pi_{\tilde{\bm{h}_1}}^{\perp}(\tilde{\bm{h}}_1^{(2)}))+\Pi_{\tilde{\bm{h}_1}^{(1)}}^{\perp}((\tilde{\bm{h}}_1^{(2)\trans}\tilde{\bm{h}}_1)\tilde{\bm{h}}_1)\|\\
&\leq\|\Pi_{\tilde{\bm{h}_1}^{(1)}}^{\perp}(\Pi_{\tilde{\bm{h}_1}}^{\perp}(\tilde{\bm{h}}_1^{(2)}))\|+\|\Pi_{\tilde{\bm{h}_1}^{(1)}}^{\perp}((\tilde{\bm{h}}_1^{(2)\trans}\tilde{\bm{h}}_1)\tilde{\bm{h}}_1)\|\\
&\leq\|\Pi_{\tilde{\bm{h}_1}}^{\perp}(\tilde{\bm{h}}_1^{(2)})\|+\|\Pi_{\tilde{\bm{h}_1}^{(1)}}^{\perp}(\tilde{\bm{h}}_1)\|\\
&\leq|\sin(\angle(\tilde{\bm{h}}_1^{(1)},\tilde{\bm{h}}_1))|+|\sin(\angle(\tilde{\bm{h}}_1^{(2)},\tilde{\bm{h}}_1))|\\
&\leq2|\sin(\max(\angle(\tilde{\bm{h}}_1^{(1)},\tilde{\bm{h}}_1),\angle(\tilde{\bm{h}}_1^{(2)},\tilde{\bm{h}}_1)))|.
\label{eq:App_UB_9}
\end{aligned}
\end{equation}
This holds for the two channels vectors $\tilde{\bm{h}}_1$ and $\tilde{\bm{h}}_2$ so that taking the maximum over all the sinus and choosing the multiplicative constant as the sum of the multiplicative constants we obtain the result of the lemma.
\end{proof}

\begin{lemma}
Let the beamformers $\bm{u}_2^{(1)}$ and $\bm{u}_2^{(2)}$ be computed at TX~$1$ and TX~$2$ respectively, then
\begin{equation}
\E\left[\log_2\left\|\bm{u}_{2}^{(2)}-\bm{u}_{2}^{(1)}\right\|^2\right]\geq \E\left[\log_2\left(C_{\mathrm{LB}} \max_{i=1,2,j=1,2}\left(\sin^2(\angle(\tilde{\bm{h}}_i^{(j)},\tilde{\bm{h}}_i))\right)\right)\right]
\end{equation}
\label{lemma_LB}
\end{lemma}
\begin{proof}
Wlog we consider that the worst CSI is obtained at TX~$1$, and we let a geni gives perfect CSI to TX~$2$ of the full channel. Then, we consider two different cases depending on whether the worst quality of an estimate is about $\tilde{\bm{h}}_1$ or $\tilde{\bm{h}}_2$ and we consider that a geni then gives perfect CSI of the channel which has not the worst acurracy.

\emph{Least accurate estimate is an estimate on $\tilde{\bm{h}}_1$:} In that case, the difference reads as
\begin{equation} 
\left\|\bm{u}_{2}^{(1)}-\bm{u}_{2}\right\|
=\left\|\frac{\Pi_{\tilde{\bm{h}}_1^{(1)}}^{\perp}(\tilde{\bm{h}}_2)}{\norm{\Pi_{\tilde{\bm{h}}_1^{(1)}}^{\perp}(\tilde{\bm{h}}_2)}}-\frac{\Pi_{\tilde{\bm{h}}_1}^{\perp}(\tilde{\bm{h}}_2)}{\norm{\Pi_{\tilde{\bm{h}}_1}^{\perp}(\tilde{\bm{h}}_2)}}\right\|
\label{eq:App_LB_1}
\end{equation}
We then rewrite $\Pi_{\tilde{\bm{h}}_1^{(1)}}^{\perp}(\tilde{\bm{h}}_2)$ as
\begin{equation} 
\Pi_{\tilde{\bm{h}}_1^{(1)}}^{\perp}(\tilde{\bm{h}}_2)=\Pi_{\tilde{\bm{h}}_1}^{\perp}\left(\Pi_{\tilde{\bm{h}}_1^{(1)}}^{\perp}(\tilde{\bm{h}}_2)\right)+\left(\tilde{\bm{h}}_1^{\trans}\Pi_{\tilde{\bm{h}}_1^{(1)}}^{\perp}(\tilde{\bm{h}}_2)\right)\tilde{\bm{h}}_1.
\label{eq:App_LB_2}
\end{equation}
which we can insert in \eqref{eq:App_LB_1} to obtain
\begin{equation}
\begin{aligned} 
\left\|\bm{u}_{2}^{(1)}-\bm{u}_{2}\right\|
&=\left\|\frac{\Pi_{\tilde{\bm{h}}_1}^{\perp}\left(\Pi_{\tilde{\bm{h}}_1^{(1)}}^{\perp}(\tilde{\bm{h}}_2)\right)+\left(\tilde{\bm{h}}_1^{\trans}\Pi_{\tilde{\bm{h}}_1^{(1)}}^{\perp}(\tilde{\bm{h}}_2)\right)\tilde{\bm{h}}_1}{\norm{\Pi_{\tilde{\bm{h}}_1^{(1)}}^{\perp}(\tilde{\bm{h}}_2)}}-\frac{\Pi_{\tilde{\bm{h}}_1}^{\perp}(\tilde{\bm{h}}_2)}{\norm{\Pi_{\tilde{\bm{h}}_1}^{\perp}(\tilde{\bm{h}}_2)}}\right\|\\
&=\left\|\frac{\Pi_{\tilde{\bm{h}}_1}^{\perp}\left(\Pi_{\tilde{\bm{h}}_1^{(1)}}^{\perp}(\tilde{\bm{h}}_2)\right)}{\norm{\Pi_{\tilde{\bm{h}}_1^{(1)}}^{\perp}(\tilde{\bm{h}}_2)}}-\frac{\Pi_{\tilde{\bm{h}}_1}^{\perp}(\tilde{\bm{h}}_2)}{\norm{\Pi_{\tilde{\bm{h}}_1}^{\perp}(\tilde{\bm{h}}_2)}}\right\|+\frac{\left|\tilde{\bm{h}}_1^{\trans}\Pi_{\tilde{\bm{h}}_1^{(1)}}^{\perp}(\tilde{\bm{h}}_2)\right|}{\norm{\Pi_{\tilde{\bm{h}}_1^{(1)}}^{\perp}(\tilde{\bm{h}}_2)}}\\
\end{aligned}
\label{eq:App_LB_2_bis}
\end{equation}
We lower bound the expression by neglecting the first term to obtain the following expression.
\begin{equation}
\begin{aligned}
\left\|\bm{u}_{2}^{(1)}-\bm{u}_{2}\right\|&\geq\frac{\left|\tilde{\bm{h}}_1^{\trans}\Pi_{\tilde{\bm{h}}_1^{(1)}}^{\perp}(\tilde{\bm{h}}_2)\right|}{\norm{\Pi_{\tilde{\bm{h}}_1^{(1)}}^{\perp}(\tilde{\bm{h}}_2)}}\\
&=\frac{\left|\left[\Pi_{\tilde{\bm{h}}_1^{(1)}}^{\perp}(\tilde{\bm{h}}_1)+(\tilde{\bm{h}}_1^{(1)\trans}\tilde{\bm{h}}_2)\tilde{\bm{h}}_1^{(1)}\right]^{\trans}\Pi_{\tilde{\bm{h}}_1^{(1)}}^{\perp}(\tilde{\bm{h}}_2)\right|}{\norm{\Pi_{\tilde{\bm{h}}_1^{(1)}}^{\perp}(\tilde{\bm{h}}_2)}}\\
&=\frac{\left|(\Pi_{\tilde{\bm{h}}_1^{(1)}}^{\perp}(\tilde{\bm{h}}_1))^{\trans}\Pi_{\tilde{\bm{h}}_1^{(1)}}^{\perp}(\tilde{\bm{h}}_2)\right|}{\norm{\Pi_{\tilde{\bm{h}}_1^{(1)}}^{\perp}(\tilde{\bm{h}}_2)}}\\
&=|\sin(\angle(\tilde{\bm{h}}_1^{(1)},\tilde{\bm{h}}_1))|\left|\cos(\angle(\Pi_{\tilde{\bm{h}}_1^{(1)}}^{\perp}(\tilde{\bm{h}}_1)),\Pi_{\tilde{\bm{h}}_1^{(1)}}^{\perp}(\tilde{\bm{h}}_2)))\right|.
\end{aligned}
\label{eq:App_LB_3}
\end{equation}
The two vectors in the cosinus are i.i.d. isotropic in the $n-1$-dimensional subspace orthogonal to $\tilde{\bm{h}}_1^{(1)}$ so that the cosinus is $\beta(1,n-2)$ distributed. It is also independent of $\tilde{\bm{h}}_1^{(1)}$ and the expectation can be taken independently.

\emph{Least accurate estimate is an estimate on $\tilde{\bm{h}}_2$:} In that case, the difference reads as
\begin{equation} 
\begin{aligned}
\left\|\bm{u}_{2}^{(1)}-\bm{u}_{2}\right\|
&=\left\|\frac{\Pi_{\tilde{\bm{h}}_1}^{\perp}(\tilde{\bm{h}}_2^{(1)})}{\norm{\Pi_{\tilde{\bm{h}}_1}^{\perp}(\tilde{\bm{h}}_2^{(1)})}}-\frac{\Pi_{\tilde{\bm{h}}_1}^{\perp}(\tilde{\bm{h}}_2)}{\norm{\Pi_{\tilde{\bm{h}}_1}^{\perp}(\tilde{\bm{h}}_2)}}\right\|\\
&=\norm{\Pi_{\tilde{\bm{h}}_1}^{\perp}(\tilde{\bm{h}}_2^{(1)})}\left\|\Pi_{\tilde{\bm{h}}_1}^{\perp}(\tilde{\bm{h}}_2^{(1)})-\frac{\norm{\Pi_{\tilde{\bm{h}}_1}^{\perp}(\tilde{\bm{h}}_2^{(1)})}}{\norm{\Pi_{\tilde{\bm{h}}_1}^{\perp}(\tilde{\bm{h}}_2)}}\Pi_{\tilde{\bm{h}}_1}^{\perp}(\tilde{\bm{h}}_2)\right\|\\
\end{aligned}
\label{eq:App_LB_4}
\end{equation}
We rewrite the term $\Pi_{\tilde{\bm{h}}_1}^{\perp}(\tilde{\bm{h}}_2^{(1)})$ as 
\begin{equation} 
\Pi_{\tilde{\bm{h}}_1}^{\perp}(\tilde{\bm{h}}_2^{(1)})=\Pi_{\tilde{\bm{h}}_1}^{\perp}\left(\Pi_{\tilde{\bm{h}}_2}^{\perp}(\tilde{\bm{h}}_2^{(1)})+(\tilde{\bm{h}}_2^{(1)\trans}\tilde{\bm{h}}_2)\tilde{\bm{h}}_2\right)
\label{eq:App_LB_5}
\end{equation}
so that the ratio of norms can be written as
\begin{equation}  
\frac{\norm{\Pi_{\tilde{\bm{h}}_1}^{\perp}(\tilde{\bm{h}}_2^{(1)})}}{\norm{\Pi_{\tilde{\bm{h}}_1}^{\perp}(\tilde{\bm{h}}_2)}}
=\frac{\norm{\Pi_{\tilde{\bm{h}}_1}^{\perp}\left(\Pi_{\tilde{\bm{h}}_2}^{\perp}(\tilde{\bm{h}}_2^{(1)})\right)+(\tilde{\bm{h}}_2^{(1)\trans}\tilde{\bm{h}}_2)\Pi_{\tilde{\bm{h}}_1}^{\perp}(\tilde{\bm{h}}_2)}}{\norm{\Pi_{\tilde{\bm{h}}_1}^{\perp}(\tilde{\bm{h}}_2)}}.
\label{eq:App_LB_6}
\end{equation}
Considering that the term~$\Pi_{\tilde{\bm{h}}_2}^{\perp}(\tilde{\bm{h}}_2^{(1)})$ is small compared to one since the accuracy increases with the SNR, we can approximate the numerator using a Taylor approximation. Thus,
\begin{equation}  
\frac{\norm{\Pi_{\tilde{\bm{h}}_1}^{\perp}(\tilde{\bm{h}}_2^{(1)})}}{\norm{\Pi_{\tilde{\bm{h}}_1}^{\perp}(\tilde{\bm{h}}_2)}}
\approx 1+\frac{\left(\Pi_{\tilde{\bm{h}}_1}^{\perp}\left(\Pi_{\tilde{\bm{h}}_2}^{\perp}(\tilde{\bm{h}}_2^{(1)})\right)\right)^{\trans}\Pi_{\tilde{\bm{h}}_1}^{\perp}(\tilde{\bm{h}}_2)}{\norm{\Pi_{\tilde{\bm{h}}_1}^{\perp}(\tilde{\bm{h}}_2)}^2}.
\label{eq:App_LB_7}
\end{equation}
Inserting \eqref{eq:App_LB_5} and \eqref{eq:App_LB_7} in \eqref{eq:App_LB_4}, we obtain
{\small\begin{equation} 
\begin{aligned}
&\left\|\bm{u}_{2}^{(1)}-\bm{u}_{2}\right\|\\
&\approx \left\|\Pi_{\tilde{\bm{h}}_1}^{\perp}\left(\Pi_{\tilde{\bm{h}}_2}^{\perp}(\tilde{\bm{h}}_2^{(1)})+(\tilde{\bm{h}}_2^{(1)\trans}\tilde{\bm{h}}_2)\tilde{\bm{h}}_2\right)-\left(1+\frac{\left(\Pi_{\tilde{\bm{h}}_1}^{\perp}\left(\Pi_{\tilde{\bm{h}}_2}^{\perp}(\tilde{\bm{h}}_2^{(1)})\right)\right)^{\trans}\Pi_{\tilde{\bm{h}}_1}^{\perp}(\tilde{\bm{h}}_2)}{\norm{\Pi_{\tilde{\bm{h}}_1}^{\perp}(\tilde{\bm{h}}_2)}^2}\right)\Pi_{\tilde{\bm{h}}_1}^{\perp}(\tilde{\bm{h}}_2)\right\|\\
&= \left\|\Pi_{\tilde{\bm{h}}_1}^{\perp}\left(\Pi_{\tilde{\bm{h}}_2}^{\perp}(\tilde{\bm{h}}_2^{(1)})\right)-\left(1-(\tilde{\bm{h}}_2^{(1)\trans}\tilde{\bm{h}}_2)+\frac{\left(\Pi_{\tilde{\bm{h}}_1}^{\perp}\left(\Pi_{\tilde{\bm{h}}_2}^{\perp}(\tilde{\bm{h}}_2^{(1)})\right)\right)^{\trans}\Pi_{\tilde{\bm{h}}_1}^{\perp}(\tilde{\bm{h}}_2)}{\norm{\Pi_{\tilde{\bm{h}}_1}^{\perp}(\tilde{\bm{h}}_2)}^2}\right)\Pi_{\tilde{\bm{h}}_1}^{\perp}(\tilde{\bm{h}}_2)\right\|\\
&\geq \left|\left\|\Pi_{\tilde{\bm{h}}_1}^{\perp}\left(\Pi_{\tilde{\bm{h}}_2}^{\perp}(\tilde{\bm{h}}_2^{(1)})\right)\right\|-\left\|\left(1-(\tilde{\bm{h}}_2^{(1)\trans}\tilde{\bm{h}}_2)+\frac{\left(\Pi_{\tilde{\bm{h}}_1}^{\perp}\left(\Pi_{\tilde{\bm{h}}_2}^{\perp}(\tilde{\bm{h}}_2^{(1)})\right)\right)^{\trans}\Pi_{\tilde{\bm{h}}_1}^{\perp}(\tilde{\bm{h}}_2)}{\norm{\Pi_{\tilde{\bm{h}}_1}^{\perp}(\tilde{\bm{h}}_2)}^2}\right)\Pi_{\tilde{\bm{h}}_1}^{\perp}(\tilde{\bm{h}}_2)\right\|\right|\\
&= \bigg||\sin(\angle(\tilde{\bm{h}}_2^{(1)},\tilde{\bm{h}}_2))||\sin(\angle(\tilde{\bm{h}}_1,\tilde{\bm{v}}))|\\
&\qquad\qquad-\left|(1-\cos(\angle(\tilde{\bm{h}}_2^{(1)},\tilde{\bm{h}}_2)))\left\|\Pi_{\tilde{\bm{h}}_1}^{\perp}(\tilde{\bm{h}}_2)\right\|+\cos(\angle(\Pi_{\tilde{\bm{h}}_1}^{\perp}(\Pi_{\tilde{\bm{h}}_2}^{\perp}(\tilde{\bm{h}}_2^{(1)})),\Pi_{\tilde{\bm{h}}_1}^{\perp}(\tilde{\bm{h}}_2)))|\norm{\Pi_{\tilde{\bm{h}}_2}^{\perp}(\tilde{\bm{h}}_2^{(1)})}\norm{\Pi_{\tilde{\bm{h}}_1}^{\perp}(\tilde{\bm{v}})}\right|\bigg|\\
&\approx \left||\sin(\angle(\tilde{\bm{h}}_2^{(1)},\tilde{\bm{h}}_2))\right|\\
&\qquad\qquad \cdot\left||\sin(\angle(\tilde{\bm{h}}_1,\tilde{\bm{v}}))|-\big|\frac{1}{2}\sin(\angle(\tilde{\bm{h}}_2^{(1)},\tilde{\bm{h}}_2))\|\Pi_{\tilde{\bm{h}}_1}^{\perp}(\tilde{\bm{h}}_2)\|-\cos(\angle(\Pi_{\tilde{\bm{h}}_1}^{\perp}(\Pi_{\tilde{\bm{h}}_2}^{\perp}(\tilde{\bm{h}}_2^{(1)})),\Pi_{\tilde{\bm{h}}_1}^{\perp}(\tilde{\bm{h}}_2)))\norm{\Pi_{\tilde{\bm{h}}_1}^{\perp}(\tilde{\bm{v}})}\big|\right|\\
\end{aligned}
\label{eq:App_LB_8}
\end{equation}}
where $\tilde{\bm{v}}\triangleq \Pi_{\tilde{\bm{h}}_2}^{\perp}(\tilde{\bm{h}}_2^{(1)})/\norm{\Pi_{\tilde{\bm{h}}_2}^{\perp}(\tilde{\bm{h}}_2^{(1)})}$ and is isotropically distributed in the $(n-1)$-space orthogonal to $\tilde{\bm{h}}_2$.

The second factor multiplying factor is non-zero with probability one and computing the expectation of this logarithm gives the result.
\end{proof}

We will now use the two previous lemmas quantifying the norm difference between the beamformers computed at the TXs to show that the MG given in the theorem is a lower and an upper bound for the MG achieved. We start from equation ($23$) in~\cite{Jindal2006} which is obtained via basic manipulations using the isotropy of the channel and is written here adapted to our notations as
\begin{equation} 
\begin{aligned}
\MG_1^{\mathrm{cZF}}&=1-\lim_{P\rightarrow \infty}\E_{\mathbf{H},\mathcal{W}}\left[\frac{\log_2\left(1+\tfrac{P}{2}\norm{\bm{h}_1}^ 2|\tilde{\bm{h}}_1^{\He}\bm{u}_2|^2\right)}{\log_2(P)}\right]\\
&= -\lim_{P\rightarrow \infty}\E_{\tilde{\bm{h}}_1,\mathcal{W}}\left[\frac{\log_2\left(|\tilde{\bm{h}}_1^{\He}\bm{u}_2|^2\right)}{\log_2(P)}\right]\\
\end{aligned}
\label{eq:App_ZF_1}
\end{equation}
where we write for simplicity $\bm{u}_2$ instead of $\bm{u}_2^{\mathrm{cZF}}$. \footnote{Note that the vector $\bm{u}_2$ is not exactly unitary due to the lack of coordination between the TXs. Indeed TX~$1$ normalizes its coefficient by $\norm{\bm{u}_2^{(1)}}$, which is a priori not equal to $\norm{\bm{u}_2^{(2)}}$. However, we consider only a number of feedback bits scaling with $\log_2(P)$ since otherwise the MG in a single TX configuration is zero and thus also in a distributed CSI configuration. Hence, the accuracy of the normalization improves with $\log_2(P)$ and the power constraint is fulfilled with an accuracy increasing in the SNR which is sufficiently good for practical systems.}

\emph{Multiplexing Gain Lower Bound:} To obtain a lower bound for the MG, we need to derive an upper bound for~$|\tilde{\bm{h}}_1^{\He}\bm{u}_2|$. We define the selection matrix $\mathbf{E}_2\triangleq\diag(\begin{bmatrix}0 &1 \end{bmatrix})$ and rewrite the interference term as: 
\begin{equation} 
\begin{aligned}
\left|\tilde{\bm{h}}_1^{\He}\bm{u}_2\right|&=\left|\tilde{\bm{h}}_1^{\He}\bm{u}_2^{(1)}+ \mathbf{E}_2(\bm{u}_{2}^{(2)}-\bm{u}_{2}^{(1)})\right|\\
&\leq \left|\tilde{\bm{h}}_1^{\He}\bm{u}_2^{(1)}\right|+\left|\tilde{\bm{h}}_1^{\He}\mathbf{E}_2(\bm{u}_{2}^{(2)}-\bm{u}_{2}^{(1)})\right|\\
&\leq\left|\tilde{\bm{h}}_1^{\He}\bm{u}_2^{(1)}\right|+\left|\tilde{\bm{h}}_1^{\He}\mathbf{E}_2 \left(\bm{u}_{2}^{(2)}-\bm{u}_{2}^{(1)}\right)\right|\\
&\leq\left|\tilde{\bm{h}}_1^{\He}\bm{u}_2^{(1)}\right|+\left\|\mathbf{E}_2 \left(\bm{u}_{2}^{(2)}-\bm{u}_{2}^{(1)}\right)\right\|\\
&\leq\left|\tilde{\bm{h}}_1^{\He}\bm{u}_2^{(1)}\right|+\left\|\bm{u}_{2}^{(2)}-\bm{u}_{2}^{(1)}\right\|.
\end{aligned}
\label{eq:App_ZF_2}
\end{equation} 
Applying Lemma~\ref{lemma_UB} we obtain the bound
\begin{equation}
\left\|\bm{u}_{2}^{(2)}-\bm{u}_{2}^{(1)}\right\|\leq  C_{\mathrm{UB}}^{(2)-(1)}\max_{i=1,2}\left(|\sin(\angle(\tilde{\bm{h}}_i^{(1)},\tilde{\bm{h}}_i))|\right)
\end{equation}
 while we can also apply the lemma for $\left\|\bm{u}_{2}^{(1)}-\bm{u}_{2}^*\right\|$ with $\bm{u}_{2}^*$ the ZF beamformer with perfect CSI, to write
\begin{equation}
\bm{u}_{2}^{(1)}=\bm{u}_{2}^*+\bm{\eta}_2^{(1)}
\end{equation}
with $\|\bm{\eta}_2^{(1)}\|\leq  C_{\mathrm{UB}}^{(1)}\max_{i=1,2}\left(|\sin(\angle(\tilde{\bm{h}}_i^{(1)},\tilde{\bm{h}}_i))|\right)$ and $C_{\mathrm{UB}}^{(1)}$ is a positive constant. Thus,
\begin{equation} 
\begin{aligned}
\left|\tilde{\bm{h}}_1^{\He}\bm{u}_2\right|
&\leq\left|\tilde{\bm{h}}_1^{\He}\bm{\eta}_2^{(1)}\right|+\left\|\bm{u}_{2}^{(2)}-\bm{u}_{2}^{(1)}\right\|\\
&\leq\|\bm{\eta}_2^{(1)}\|+\left\|\bm{u}_{2}^{(2)}-\bm{u}_{2}^{(1)}\right\|\\
&\leq C_{\mathrm{UB}} \max_{i=1,2;j=1,2}\left(|\sin(\angle(\tilde{\bm{h}}_i^{(j)},\tilde{\bm{h}}_i))|\right)
\end{aligned}
\label{eq:App_ZF_3}
\end{equation}
with $C_{\mathrm{UB}}\triangleq C_{\mathrm{UB}}^{(2)-(1)}+C_{\mathrm{UB}}^{(1)}$. Then, inserting~\eqref{eq:App_ZF_3} in \eqref{eq:App_ZF_1} gives
{\begin{equation} 
\begin{aligned}
\MG_1^{\mathrm{cZF}}&=-\lim_{P\rightarrow \infty}\E_{\tilde{\bm{h}}_1,\mathcal{W}}\left[\frac{\log_2\left(|\tilde{\bm{h}}_1^{\He}\bm{u}_2|^2\right)}{\log_2(P)}\right]\\
&\geq -\lim_{P\rightarrow \infty}\E_{\tilde{\bm{h}}_1,\mathcal{W}}\left[\frac{\log_2\left(\left|\tilde{\bm{h}}_1^{\He}\bm{u}_2^{(1)}\right|+\left\|\bm{u}_{2}^{(2)}-\bm{u}_{2}^{(1)}\right\|^ 2\right)}{\log_2(P)}\right]\\
&\geq -\lim_{P\rightarrow \infty}\E_{\tilde{\bm{h}}_1,\mathcal{W}}\left[\frac{\log_2\left((C_{\mathrm{UB}}+C_{\mathrm{UB}}^{(1)}) \max_{i=1,2;j=1,2}\left(\sin^2(\angle(\tilde{\bm{h}}_i^{(j)},\tilde{\bm{h}}_i))\right)\right)}{\log_2(P)}\right].
\end{aligned}
\label{eq:App_ZF_4}
\end{equation}}
From \eqref{eq:App_ZF_4} we can then use Jensen's inequality and Proposition~\ref{App_distorsion} in Appendix~\ref{se:Appendix_RVQ} to obtain the lower bound from the theorem.
\begin{equation} 
\begin{aligned}
\MG_1^{\mathrm{cZF}}
&\geq -\lim_{P\rightarrow \infty}\frac{\log_2\left(C_{\mathrm{UB}}'\E_{\tilde{\bm{h}}_1,\mathcal{W}}\left[ \max_{i=1,2;j=1,2}\left(\sin^2(\angle(\tilde{\bm{h}}_i^{(j)},\tilde{\bm{h}}_i))\right)\right]\right)}{\log_2(P)}\\
&\geq -\lim_{P\rightarrow \infty}\frac{\log_2\left(C_{\mathrm{UB}}'\frac{\Gamma(\frac{1}{n-1})}{n-1}c_{n}^{-1/(n-1)}2^{-\min_{i,j}(B_i^{(j)})/(n-1)}(1+o(1))\right)}{\log_2(P)}\\
&= \lim_{P\rightarrow \infty}\frac{\min_{i,j}B_i^{(j)}}{(n-1)\log_2(P)}\\
&=\min_{i,j}\alpha_i^{(j)}.
\end{aligned}
\label{eq:App_ZF_5}
\end{equation}

\emph{Multiplexing Gain Upper Bound:} We now derive an upper bound for the MG, which means a lower bound for the interference. We consider wlog that TX~$1$ has the worst channel estimate and let a geni give perfect CSI to TX~$2$. We can now use Lemma~\ref{lemma_LB} to derive an upper bound for the interference term.
\begin{equation} 
\begin{aligned}
&\E\left[\log_2\left(|\tilde{\bm{h}}_1^{\He}\bm{u}_2|^2\right)\right]\\
=&\E\left[\log_2\left(\left|\tilde{\bm{h}}_1^{\He}\left(\bm{u}_2^*+ \mathbf{E}_1(\bm{u}_{2}^{(1)}-\bm{u}_{2}^*)\right)\right|^2\right)\right]\\
=&\E\left[\log_2\left(\left|\cos(\angle(\tilde{\bm{h}}_1^{\He}\mathbf{E}_1,\bm{u}_{2}^{(1)}-\bm{u}_{2}^*)) \norm{\tilde{\bm{h}}_1^{\He}\mathbf{E}_1}\norm{\bm{u}_{2}^{(1)}-\bm{u}_{2}^*}\right|^2\right)\right]\\
\geq&\E\!\left[\log_2\left(\left|\cos(\angle(\tilde{\bm{h}}_1^{\He}\mathbf{E}_1\!,\bm{u}_{2}^{(1)}\!-\!\bm{u}_{2}^*)) \norm{\tilde{\bm{h}}_1^{\He}\mathbf{E}_1}\right|^2\right)+\log_2\left(C_{\mathrm{LB}}^{(1)} \!\max_{i=1,2,j=1,2}\!\left(\sin^2(\angle(\tilde{\bm{h}}_i^{(j)}\!,\tilde{\bm{h}}_i))\!\right)\right)\right]\!\\
\end{aligned}
\label{eq:App_ZF_6}
\end{equation} 
Inserting \eqref{eq:App_ZF_6} into the MG expression \eqref{eq:App_ZF_1} yields
\begin{equation} 
\begin{aligned}
\MG_1^{\mathrm{ZF}} 
&= -\lim_{P\rightarrow \infty}\frac{\E_{\tilde{\bm{h}}_1,\mathcal{W}}\left[\log_2\left(|\tilde{\bm{h}}_1^{\He}\bm{u}_2|^2\right)\right]}{\log_2(P)}\\
&\leq -\lim_{P\rightarrow \infty}\frac{\E_{\tilde{\bm{h}}_1,\mathcal{W}}\left[\log_2\left(C_{\mathrm{LB}}^{(1)} \max_{i=1,2,j=1,2}\left(\sin^2(\angle(\tilde{\bm{h}}_i^{(j)},\tilde{\bm{h}}_i))\right)\right)+O(1)\right]}{\log_2(P)}\\
&= -\lim_{P\rightarrow \infty}\frac{\E_{\tilde{\bm{h}}_1,\mathcal{W}}\left[\log_2\left(\max_{i=1,2,j=1,2}\left(\sin^2(\angle(\tilde{\bm{h}}_i^{(j)},\tilde{\bm{h}}_i))\right)\right)\right]}{\log_2(P)}\\
&\leq \lim_{P\rightarrow \infty}\frac{\min_{i=1,2,j=1,2}B_i^{(j)}+\log_2(c_n)+\log_2(e)}{(n-1)\log_2(P)}\\
&= \min_{i=1,2,j=1,2}\alpha_i^{(j)}
\end{aligned}
\label{eq:App_ZF_7}
\end{equation}
and we have used Proposition~\ref{App_log_distorsion} in Appendix~\ref{se:Appendix_RVQ} to bound the expectation of the logarithm of the sinus.
\end{proof}

\bibliographystyle{IEEEtran}
\bibliography{Literatur}
\end{document}